\documentclass[format=acmsmall, review=false]{acmart}
\pdfoutput=1
\usepackage{acm-ec-24-proc}
\usepackage{booktabs} 
\usepackage[ruled]{algorithm2e} 
\usepackage{natbib} 

\SetAlFnt{\small}
\SetAlCapFnt{\small}
\SetAlCapNameFnt{\small}
\SetAlCapHSkip{0pt}
\IncMargin{-\parindent}

\setcitestyle{authoryear}



\usepackage[T1]{fontenc}
\usepackage{wrapfig}

\usepackage[utf8]{inputenc}
\usepackage{amsthm}
\usepackage{amsmath}
\usepackage{breqn}
\usepackage{hyperref}
\usepackage[T1]{fontenc}
\usepackage{wrapfig}
\usepackage{enumitem}

\usepackage{graphicx}

\allowdisplaybreaks

\newenvironment{hproof}{%
  \renewcommand{\proofname}{Proof Sketch}\proof}{\endproof}

\newlist{pfparts}{description}{1}
\setlist[pfparts,1]{%
  font=\normalfont\textsf,
  itemindent=2pt,
  wide,
  itemsep=0pt,topsep=2pt,
  labelsep=0.75ex
}

\newcommand{\ashish}[1]{{\color{blue} [AG: {#1}]}}

\newcommand{\mohak}[1]{{\color{red} [MG: {#1}]}}

\usepackage{tikz}
\usepackage{pgfplots}\pgfplotsset{compat=1.9}
\usepackage{subcaption}

\title{Augmenting Batch Exchanges with Constant Function Market Makers
}

\acmYear{2024}\copyrightyear{2024}
\acmConference[EC '24]{Conference on Economics and Computation}{July 8--11, 2024}{New Haven, CT, USA}
\acmBooktitle{Conference on Economics and Computation (EC '24), July 8--11, 2024, New Haven, CT, USA}
\acmDOI{10.1145/3670865.3673569}
\acmISBN{979-8-4007-0704-9/24/07}

\newcommand{\geoffemail}{geoff.ramseyer@cs.stanford.edu}
\newcommand{\mohakemail}{mohakg@stanford.edu}
\newcommand{\ashishemail}{ashishg@stanford.edu}
\newcommand{\dmemail}{}

\author{Geoffrey Ramseyer*}
\affiliation{%
  \institution{Stanford University}
  \city{Stanford}
  \state{California}
  \postcode{94305}
  \country{USA}}
\email{\geoffemail}

\author{Mohak Goyal*}
\affiliation{%
  \institution{Stanford University}
  \city{Stanford}
  \state{California}
  \postcode{94305}
  \country{USA}
}
\email{\mohakemail}

\author{Ashish Goel}
\affiliation{%
  \institution{Stanford University}
  \city{Stanford}
  \state{California}
  \postcode{94305}
  \country{USA}}
\email{\ashishemail}

\author{David Mazières}
\affiliation{%
  \institution{Stanford University}
  \city{Stanford}
  \state{California}
  \postcode{94305}
  \country{USA}}
\email{\dmemail}
\thanks{* denotes equal contribution.}

\newcommand\asset[1]{\begingroup #1 \endgroup}
\newcommand\price[1]{\begingroup {p_{\asset{#1}}} \endgroup}

\newcommand{\assetset}{\begingroup \mathcal{A} \endgroup}

\newcommand{\agentset}{\mathcal{I}}
\newcommand{\cvxcfmmset}{\mathcal{J}}
\newcommand{\cvxloset}{\mathcal{K}}
\newcommand{\cvxboth}{\cvxcfmmset \cup \cvxloset}

\newtheorem{theorem}{Theorem}[section]
\newtheorem{corollary}[theorem]{Corollary}
\newtheorem{lemma}[theorem]{Lemma}
\newtheorem*{lemma*}{Lemma}
\newtheorem{example}[theorem]{Example}
\newtheorem*{example*}{Example}
\newtheorem{definition}[theorem]{Definition}
\newtheorem*{theorem*}{Theorem}
\newtheorem{observation}{Observation}
\newtheorem{assumption}{Assumption}
\newtheorem*{proposition*}{Proposition}
\newtheorem{axiom}{Axiom}

\newtheorem{claim}{Claim}

\begin{abstract}

Batch auctions are a classical market microstructure, acclaimed for their fairness properties, and have received renewed interest in the context of blockchain-based financial systems. Constant function market makers (CFMMs) are another market design innovation praised for their computational simplicity. 
Liquidity provision in batch exchanges is an important problem, and CFMMs have recently shown promise in being useful within batch exchanges. Different real-world implementations have used fundamentally different approaches towards integrating CFMMs in batch exchanges, and there is a lack of formal understanding of the trade-offs of different design choices.

We first provide a minimal set of axioms that capture the well-accepted rules of batch exchanges and CFMMs. For batch exchanges, these are asset conservation, uniform prices, and a best response for limit orders. For CFMMs, our axiom is that their trading function is non-decreasing. Many market solutions may satisfy all our axioms. We then describe several economically desirable properties of market solutions. These include Pareto optimality for limit orders, price coherence of CFMMs (as a defence against cyclic arbitrage), joint price discovery for CFMMs (as a defence against parallel running), path independence, and a locally computable trade response of the CFMMs (to provide them with a predictable trade size given a market price). For market solutions satisfying all our axioms, we show fundamental conflicts between some pairs of desirable properties. Most notably, a batch exchange cannot simultaneously guarantee `Pareto optimality' for the limit orders and any of `price coherence' or `locally computable response' for the CFMMs. We then provide two ways of integrating CFMMs in batch exchanges, which attain different subsets of these properties. 

We further provide a convex program for computing Arrow-Debreu exchange market equilibria when all agents have weak gross substitute (WGS) demand functions on two assets -- this program extends the literature on Arrow-Debreu exchange markets and may be of independent interest.

\end{abstract}

\begin{document}

\begin{titlepage}

\maketitle

\end{titlepage}

\section{Introduction}

A crucial component of any economic system is a structure to facilitate the exchange
of assets. 
%
Batch auctions (sometimes referred to as call auctions or call markets) are a market mechanism that accumulates trade orders over time.
At some frequency,
the exchange operator computes a uniform clearing price (the ``batch price'') and settles all trades that are possible at that price. 
Batch auctions have been a prominent market mechanism in academic literature, with models suggesting that they can lead to better price discovery and reduce intermediation costs by enabling
traders to trade with each other directly at the same time\cite{madhavan1992trading, cohen2001electronic, economides2001electronic, schwartz2012electronic}. 

There has been renewed interest in batch auctions following the work of \citet{budish2015high}, who propose using batch auctions to address the problem of competition on speed rather than on price in continuous double auctions (CDA).
Making every trade in a batch at the same price eliminates a large fraction of front-running\footnote{Front-running is the practice of making 
a trade based on advance knowledge of an upcoming order. For example, a trader can front-run a buy order for an asset by buying some of that asset, driving up its price,
and then reselling the asset to the buy order at the higher price. Such practices
are partly curtailed by regulation in several markets but are still observed in stock trading \cite{manahov2016front} and are rampant in blockchain-based exchanges \cite{daian2020flash}.}
 opportunities \cite{budish2015high}. 
Critics of batch auctions argue that they increase price uncertainty and reduce market liquidity, as liquidity providers who gain a ``speed tax'' in CDAs have no incentive to participate in batch auctions
\cite{mizuta2016investigation, lee2020frequent, dorre2020}.

While the applicability of batch auctions to traditional exchanges is still the subject of debate and regulatory considerations,
batch exchanges are especially attractive for cryptocurrencies since blockchains inherently register trades in discrete-time `blocks'. Such systems 
have already been deployed \cite{cowswap,zswap}. 
 Some batch exchanges, e.g., \cite{ramseyer2021speedex,cowswap}, process a large number of assets in one batch, instead of just two assets, by computing in every batch
a set of arbitrage-free prices between every asset pair. 
This reduces the problem of liquidity fragmentation, which is especially difficult in modern
blockchains \cite{lehar2022liquidity}. 
Furthermore, this allows users to directly trade from any asset to any other without holding some intermediate asset (such as USD), unlike the exchanges that facilitate trades only in pairs of assets.
However, the computation of market equilibrium is substantially more difficult when many assets are traded in a batch rather than just two.

Another recently developed tool for improving exchange performance on blockchains are Constant Function Market Makers, a class of automated market-makers. 
 Liquidity providers deposit capital into a CFMM, and the CFMM constantly offers trades according to a predefined trading strategy.
This strategy is specified by a ``trading function'' $f(x)$ of its asset reserves $x$ (henceforth its ``state''),
and the CFMM accepts a trade if the trading function's value does not decrease on making the trade. 
We describe CFMMs in detail in \S\ref{sec:cfmm}.
CFMM-based decentralized exchanges (DEX) such as Uniswap \cite{uniswapv2,uniswapv3} and Curve \cite{egorov2019stableswap}
are some of the largest on-chain trading platforms.


Since CFMMs are widely used in practice as relatively simple yet effective liquidity provision tools in DeFi, we investigate the possibility of designing batch exchanges which support CFMMs for liquidity.
Previous works have chosen fundamentally different methods
for mediating the interaction \cite{cowswap,canidio2023batching}, and it continues to be a problem of great interest for practitioners. We study the trade-offs between different desirable properties of batch exchanges that support CFMMs. Importantly, we show that many natural desirable properties are not simultaneously satisfiable, and therefore, we study the trade-offs of different mechanisms. 

\subsection{Preliminaries}
%
%
%
\subsubsection{Constant Function Market Makers} \label{sec:cfmm}

A CFMM is a market-making strategy
parameterized by a trading function $f: \mathbb{R}^n_{\geq0} \rightarrow \mathbb{R}_{>0},$ where $n$ is the number of assets it trades in. 
At any time,
the CFMM owns non-negative amounts of some assets (its ``reserves'') provided by deposits from investors participating in liquidity provision (so-called ``liquidity providers'').
A CFMM with reserves $x$ and function $f$ accepts any trade that results in reserves $z$ 
so long as $f(z) \geq f(x)$.  

\begin{assumption}
\label{ass:cfmm}
 CFMM trading functions are strictly quasi-concave, differentiable, nonnegative, and nondecreasing (in every coordinate)  on the positive orthant.
  \footnote{We can relax strict quasi-concavity to quasi-concavity and differentiability to continuity for many of our results.
However, we use the strong form of the assumption in \S\ref{sec:axioms} for ease of exposition.}
\end{assumption}

This assumption is standard in the literature and is important for the CFMM to be not \textit{obviously} exploitable \cite{angeris2020improved,schlegel2022axioms}. Quasi-concavity ensures that the prices are monotonic in the asset amounts.

%
%
%
The gradient of the trading function gives the price for a trade of infinitesimal size.



\begin{definition}[Spot Price]
\label{defn:spot}
The {\it spot price} from asset $\asset{A}$ to asset $\asset{B}$ for a CFMM with trading function $f$ at state $\hat{x}$ is 
$\frac{\partial f}{\partial x_A}(\hat{x})/\frac{\partial f}{\partial x_B}(\hat{x})$.  

\end{definition}


\begin{figure}
\begin{center}
  \begin{minipage}[c]{0.38\textwidth}
\includegraphics[width=\textwidth]{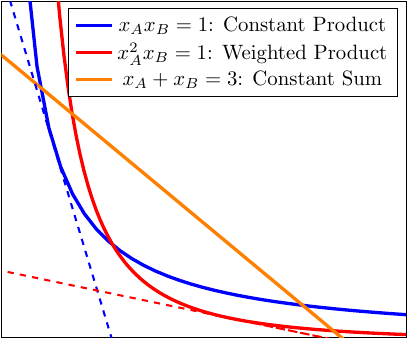}
  \end{minipage}\hfill
  \begin{minipage}[c]{0.6\textwidth}
    \caption{
    Examples of level sets of commonly studied CFMM trading functions. X-axis and Y-axis denote the amounts of assets A and B in the liquidity pool of the CFMM, respectively. Examples of spot prices are illustrated by dashed lines. Where the trading function is differentiable, the negative slope of the tangent at a point gives the spot price (Definition~\ref{defn:spot}).}
        \label{fig:cfmm_examples}
  \end{minipage}
\end{center}
\end{figure}

We illustrate some examples of widely used CFMM trading functions in Figure~\ref{fig:cfmm_examples} and give more examples in Appendix~\ref{app:cfmm_examples}.


\subsubsection{Arrow Debreu Exchange Markets}
\label{sec:AD}

An Arrow Debreu market \cite{arrow1954existence} is used to model a pure exchange market
, i.e., a market where a set of assets $\assetset$ are traded without a designated numeraire or ``money''. Assets are fungible, divisible and freely disposable. Agents have quasi-concave and non-satiating utilities for the bundle of assets they consume and have an initial endowment of assets. All trades happen at the same prices $\lbrace \price{A} >0\rbrace_{\asset{A}\in\assetset}$, and asset amounts remain conserved.

  The following are examples of order types that a batch exchange may support.%

\begin{definition}[Limit Sell Order]
\label{ex:limitsell}

A limit sell order $ (\asset{A} ,\asset{B} ,k \in \mathbb{R}_{> 0},r\in \mathbb{R}_{\geq 0})$, 
is an order to sell up to $k$ units of asset $\asset{A}$ for as many units of asset $B$ as possible,
subject to receiving a minimum price (the ``limit price'') of $r$ $\asset{B}$ per $\asset{A}$.

\end{definition}

As a convention in the literature, limit sell orders are thought of as corresponding to the utility function: $u(x) = r x_A + x_B$. We also adopt this convention in this paper. Observe that the bundle received by the limit order in a batch maximizes this utility function.

\begin{definition}[Market Sell Order]
Limit sell orders with limit price $0$ are \textit{market sell} orders.
\end{definition}


\begin{definition}[Limit Buy Order]
\label{ex:limitbuy}
A limit buy order is an order to purchase up to $k$ units of asset $A$ by selling as few units of asset $\asset{B}$ as possible,
subject to a maximum price of $r$ $\asset{B}$ per unit $\asset{A}$.
\end{definition}

Limit buy order can be seen as corresponding to the utility function: $u(x) = r x_{\asset{A}} + \min(k, x_{\asset{B}})$.



\subsection{System Model} \label{sec:model}
In this work, we design a batch trading system where multiple CFMMs can interoperate and provide liquidity to limit sell orders. 
The design specifications of our model are as follows:
\begin{enumerate}[leftmargin = 18pt]
\item A limit sell order\footnote{Limit buy orders correspond to additively separable, piecewise-linear concave utility functions and lead to PPAD-hardness of equilibrium computation, as shown by Chen et al. \cite{chen2009settling} -- therefore, we do not support it. This, however, does not make any significant restriction since a trader who wants to buy A in exchange for B can instead place an order to sell B for A.} can be submitted or removed any time before a cutoff for each next batch.
    \item A CFMM participating in the batch exchange must submit its ``state and trading function'' to the exchange before a cutoff time for each next batch.
    \item Between consecutive batches, the CFMMs may also be available for their standalone operation, but need to be unchanged after submitting their state and trading function to the batch exchange till the batch is executed.
\end{enumerate}


We do not model the fee that the batch exchange operator and the CFMM charge. 
While fees are essential to incentivize the liquidity providers to participate, we believe that our axiomatic framework is important also in the presence of trading fees.

\subsection{Our Results}

To our knowledge, this is the first paper to do an axiomatic and computational study of the important problem of augmenting batch exchanges with CFMMs. We first give a minimal set of axioms for batch exchanges incorporating CFMMs. The details are in \S\ref{sec:axioms}.


\subsubsection{Axioms}~

  We require that a batch exchange neither burn nor mint any asset -- \textit{asset amounts must be conserved} (Axiom~\ref{axiom:conservation}). We also impose the core fairness property of batch exchanges that asset prices exist in equilibrium, and no market participant receives a better trade than that implied by these prices (Axiom~\ref{axiom:batch_valuations}). 
  Further, all limit orders receive a trade which is a ``best response'' to the batch prices (Axiom~\ref{axiom:br}). 
The axiom about CFMMs is their basic design principle that their trading function should not decrease due to trade (Axiom~\ref{axiom:cfmm_nondecreasing}). 

Axioms~\ref{axiom:conservation},~\ref{axiom:batch_valuations}, and~\ref{axiom:br} are simply an articulation of the core properties of batch exchanges which we believe are well-accepted rules for designing exchange markets. Axiom~\ref{axiom:cfmm_nondecreasing} is a basic guarantee required by all CFMM deployments. These axioms are \textit{minimal} and do not impose a particular solution. They allow a rich set of solution concepts with different economically useful properties, which is this paper's core subject of study. 

We define a \textit{market equilibrium} (Definition~\ref{def:market_eq}) as a solution which satisfies Axioms~\ref{axiom:conservation},~\ref{axiom:batch_valuations},~\ref{axiom:br}, and~\ref{axiom:cfmm_nondecreasing}.
%
%
At least one market equilibrium always exists, for example, where the CFMMs do not make any trade and the set of limit sell orders trade as per a standard Arrow Debreu market.


\subsubsection{Desirable Properties of Batch Exchange Market Equilibria} \label{sec:des}
We now define some desirable properties of market equilibria. These properties then guide us in designing algorithms for finding economically useful market equilibria. 
\begin{enumerate}[leftmargin = 18pt]
    \item 
 The first property is \emph{Pareto optimality} for the limit sell orders.

\item 
Another property is \emph{price coherence} (PC) of a group of CFMMs post-batch (Definition~\ref{def:pc}). 
PC is a necessary and sufficient condition for the participating CFMMs to be in a \emph{cyclic arbitrage-free} state (Definition~\ref{def:ca}) after the batch executes. 


A weaker condition than PC is \emph{preservation of price coherence} (PPC), under which a group of participating CFMMs must be price coherent post-batch \textit{if} they were price coherent pre-batch.

 \item 
We also consider \emph{joint price discovery} (JPD) (Definition~\ref{def:jpd}) -- a property strictly stronger than PC. JPD ensures that the post-batch spot prices are the same as the batch prices (that is, the CFMM's `learn' the prices `discovered' in the batch exchange they participate in). 
JPD eradicates a form of arbitrage we call \textit{parallel running} (Definition~\ref{def:pr}). 
 \item 
Another property is \textit{locally computable response} (LCR) for the CFMMs (Definition~\ref{def:lcr}). In LCR, the trade made by a CFMM is a deterministic function of only its trading function, pre-batch state, and batch prices. LCR provides predictability to the liquidity providers and can help them do a better risk-profit analysis before committing to a market-making strategy.
\end{enumerate}

We also discuss a property -- \textit{path independence} (Definition~\ref{def:pi}) in Appendix \S~\ref{sec:pi} motivated from the standalone operation of CFMMs for batches with a single limit order.
\subsubsection{Achievability of Desirable Properties of Batch Exchange Market Equilibria} \label{sec:des-eq}
We start with two key impossibility results regarding the desirable properties of market solutions.
\begin{theorem} \label{thm:pc-po}
A batch exchange cannot simultaneously guarantee Pareto optimality for limit orders and preservation of price coherence (PPC) for CFMMs.
\end{theorem}

\begin{theorem} \label{thm:lcr-po}
A batch exchange cannot simultaneously guarantee Pareto optimality for limit orders and a locally computable response (LCR) for CFMMs.
\end{theorem}

Both proofs are via carefully designed counter-examples and are given in \S~\ref{sec:desirables}.

Recall that PPC is weaker than PC, which in turn is weaker than JPD, therefore Theorem~\ref{thm:pc-po} also applies to PC and JPD. 
Also, recall that Pareto optimality is defined for the limit orders, whereas PPC and LCR are intended to protect the CFMMs from arbitrage and trade uncertainty. These results shed light on an inherent tension between the interests of the CFMM and those of the limit orders. They also illustrate that the problem we study in this paper is non-trivial and nuanced.

We now turn to designing algorithms that find market equilibria with some of the desirable properties given above. Specifically we mainly study two natural algorithms.

\begin{enumerate}
    \item Once can view a CFMM's trading function as a pseudo-utility function and give each CFMM a bundle which maximizes this pseudo-utility. We call this ``Trading Rule U'' (U for utility) (Definition~\ref{def:tru}), and discuss it in more detail in \S~\ref{sec:pos}.
    \item The second approach is applicable only when each CFMM trades in two assets, which is the most important case in theory and practice. Here we maximize the CFMM's price-weighted trade volume at the batch prices. We call this "Trading Rule S" (S for strict, since the CFMM's trading function is strictly preserved under this rule) (Definition~\ref{def:trs}).
\end{enumerate}

Observe that both Trading Rules U and S satisfy LCR. These simple rules have further interesting economic properties, which we briefly mention here and discuss in detail in \S\ref{sec:pos}.

\begin{proposition}\label{prop:u-jpd}
    A batch exchange has JPD if and only if it uses Trading Rule U for all CFMMs.  
\end{proposition}

Trading Rule S, in general, does not satisfy PPC. However, this negative result is bypassed in batches with only ``Concentrated Liquidity Constant Product'' (CLCP) CFMMs (Definition~\ref{def:clcp}), which is a class including the constant-sum and constant-product CFMMs.

\begin{theorem} \label{thm:univ3-ppc}
CLCP is the unique class of CFMMs such that if all CFMMs in a batch belong to this class, then the market equilibria attained by implementing Trading Rule S for all CFMMs satisfy PPC.  
\end{theorem}
This result uncovers a very interesting and useful fact about the CLCP class of CFMMs, which is used extensively in practice and theory. While CLCP trading functions are hailed for their computational simplicity and universality \cite{uniswapv3}, we establish the surprising fact here that they are also particularly well-suited for integration with batch exchanges while using a very natural trading rule, which maximizes their trade volume.

\begin{figure}
\begin{center}
  \begin{minipage}[c]{0.5\textwidth}
\includegraphics[width=0.95\textwidth]{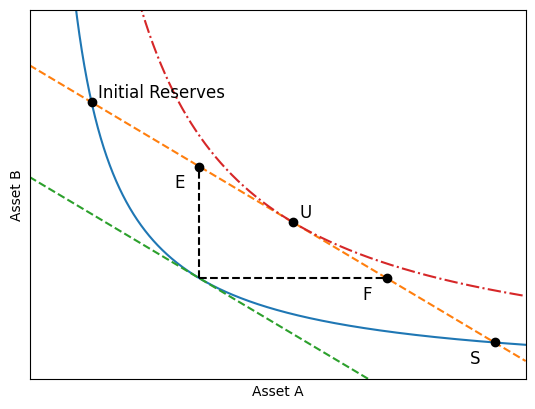}
  \end{minipage}\hfill
  \begin{minipage}[c]{0.5\textwidth}
    \caption{ 
   Examples of LCR CFMM trading rules. The axes are the CFMM's reserves. The blue curve is a level curve of the CFMM trading function on which the initial state lies.
    The slope of the tangent to the level curve denotes the CFMM's spot price.
    The slope of the dotted orange line
    is the batch price for this example. Here, the green and the orange lines have the same slopes. \\
    Trading Rules S and U are per Definitions~\ref{def:trs} and~\ref{def:tru}.
    Trading Rules E and F are inspired by the rebalancing strategy of \citet{milionis2022automated} and 
    are defined later in \S\ref{sec:beyond-market}. The line segment between points E and F corresponds to the class of Strict-Surplus Trading Rules (Definition~\ref{def:sstr}).}
        \label{fig:rules}
  \end{minipage}
\end{center}
\end{figure}

\subsubsection{Results on Equilibrium Computation}

 Convex programs for computing market equilibria have led to an improved structural
understanding of markets (such as the much-celebrated result of \citet{eisenberg1959consensus}).
We develop a convex program (in \S \ref{sec:cvx}) for computing equilibria for markets with limit sell orders and 2-asset CFMMs that have a WGS demand under any given LCR CFMM trading rule. This program may be of independent interest -- it makes progress on the open question of designing convex programs for nonlinear utility functions in Arrow Debreu exchange markets. Our program 
handles the case where agent utility functions are arbitrary quasi-concave functions of two assets, subject to the constraint that each agent's demand function is WGS.

Our convex program is inspired by that of \citet{devanur2016rational} for the case of linear utility functions.
The proof is quite technical, but the intuition is easy to state -- we develop a viewpoint from which CFMMs appear as an uncountable collection of infinitesimal limit sell offers.
Unlike the analyses in \citet{devanur2016rational} based on Lagrange multipliers, we have to go back to earlier techniques
and directly apply Kakutani's fixed point theorem.

When the density of this infinite collection of limit offers is a rational linear function, we prove that our convex program has rational solutions. This includes many classes of commonly used CFMMs, including the constant product CFMM. 
On the flip side, we show that 
some CFMMs, for example, those based on the Logarithmic Market Scoring Rule (LMSR) \cite{hanson2007logarithmic}, can force batches to admit only irrational equilibria.

\subsubsection{Solution Concepts Beyond Market Equilibria}


A natural question is whether we can attain more of the desirable properties of \S\ref{sec:des} simultaneously if we allow the CFMMs to do a post-processing step after they trade in the batch. The answer is yes, but with a caveat.
If the CFMMs have the option to alter their state post-batch, then they can together attain a state of price coherence, but this will have to violate the axiom of uniform prices (Axiom~\ref{axiom:batch_valuations}). Further, we give a class of LCR trading rules -- Strict-Surplus Trading Rules (Definition~\ref{def:sstr}) -- on 2-asset CFMMs with the following property:

\textit{In any market equilibrium where all CFMMs' trades are per a Strict-Surplus Trading Rule, each CFMM has a surplus-capturing post-processing step, upon which they together attain JPD} (Theorem~\ref{thm:brf-postprocess}). 

This result opens up a new dimension to the solution space for our problem. 
Such multiple-step solutions concepts may provide a viable design in practice, but their economic implications must be further scrutinised in future research.
More details on these results are in the Appendix \S\ref{sec:beyond-market}.

\subsection{Related Work}
\label{sec:literature}


Contemporaneous to our work (shortly after our online version appeared), \citet{canidio2023batching} describe a single CFMM which trades in two assets and executes trades in batches. Their CFMM maximizes its trading function, akin to our Trading Rule U. They show that such a design can eliminate arbitrage profits when there is competition between multiple arbitrageurs. Our work differs in that we develop an axiomatic framework to study batch exchanges where multiple CFMM can interoperate, and the exchange trades in multiple assets rather than just two. 

Most closely related to our work in the blockchain space is the work of the company CoWSwap (Coincidence of Wants-- Swap, formerly known as Gnosis); their implementation details are provided in \citet{dfusionmodel}. As with our work, they study batches that incorporate CFMMs and trade in multiple assets. They take an optimization approach with various objectives, such as maximizing trade volume or maximizing trader utility. Their solutions are based on mixed integer programs and do not have runtime guarantees. In contrast, we take an axiomatic approach and propose polynomial time solutions for finding market equilibria with certain desirable properties. 

\citet{kyle1985continuous} studies a model of a two-asset batch auction, in which a market maker a priori declares a pricing rule that maps the excess demand for the asset to a price. 
In contrast, we consider multi-asset batches and the case of potentially many market makers with their own CFMM trading functions.


Another blockchain-based system ZSwap designed by \citet{zswap} considers batches of two assets only. They aggregate \textit{market} orders in a batch and execute the excess demand with a CFMM per its constant function trading rule, subject to a maximum slippage tolerance. 

\citet{galal2021publicly}, \citet{constantinides2021block}, and \citet{memar2023practical} develop blockchain protocols to securely run batch auctions on-chain.


Automated market-making strategies have been studied both in a cryptocurrency context \cite{angeris2020improved} 
and in traditional exchanges \cite{glosten1985bid,gerig2010automated,kyle1985continuous,othman2013practical}. CFMMs form a subclass of automated market-making.
There has been extensive study on how the design of a CFMM trading function interacts with the economic incentives of those who invest in it
\cite{evans2021optimal,capponi2021adoption,neuder2021strategic,fan2022differential,cartea2022decentralised}. 
\citet{schlegel2022axioms,frongillo2023axiomatic,park2022conceptual} study the axiomatization of meaningful CFMM trading functions. \citet{goyal2022finding,milionis2023myersonian} develop frameworks for finding `optimal' CFMM trading functions.  However, this direction of work is orthogonal to the subject of study of this paper. We are interested only in the contract the CFMM enters into with the batch exchange operator.
%


\citet{budish2015high,budish2014implementation,aquilina2022quantifying} study the economic performance of batch auctions between pairs of assets. 
The well-studied model of \citet{arrow1954existence} forms the basis for multi-asset, pure exchange batch trading implementations 
\cite{speedexpricingcap,dfusionmodel}.
There are many classes of algorithms for (approximately or exactly) computing equilibria in Arrow-Debreu exchange markets, including iterative methods (or T\^atonnement)
\cite{cole2008fast,codenotti2005market,codenotti2005polynomial,bei2019ascending,garg2019auction}, 
convex programs for the case of linear utilities \cite{devanur2016rational, jain2007polynomial,nenakov1983one,cornet1989linear}, convex program for Cobb-Douglas utilities \cite{curtis1985finite}, 
and combinatorial algorithms \cite{jain2003approximating,devanur2003improved,duan2015combinatorial,garg2019strongly}. 

Most closely related to our work in this line is the convex program of \citet{devanur2016rational}, which also gives concise proof of the existence and rationality of equilibria. We generalize their program beyond linear utilities to 2-asset WGS utility functions. The convex program of \citet{nenakov1983one} (which was also given by \citet{jain2007polynomial}) also goes beyond linear utilities but in an incomparable manner from ours. Their program can handle constant elasticity of substitution (CES)\footnote{A CES utility has the form $u(x) = (\sum_{A \in \mathcal{A}} (w_A x_A)^{\rho})^{1/\rho}$ for nonnegative constants $\{w_A\}_{A \in \mathcal{A}}$. 
In the limit $\rho \rightarrow 0^+,$ we get the Cobb-Douglas utilities of the form $u(x) = \prod_{A \in \mathcal{A}} x_A^{w_A} $ for $\sum_{A \in \mathcal{A}} w_A = 1.$ The Cobb-Douglas function is widely used as the trading function of many common CFMMs \cite{balancer,uniswapv2}.
} utilities and Cobb-Douglas utilities on any number of assets but not the entire class of WGS utility functions. 
Whereas our program handles all WGS utility functions when each agent has utility for two assets.

\section{Axioms of Batch Exchanges with CFMMs}
\label{sec:axioms}
In this section, we describe our set of axioms that specify minimal guarantees that a market solution must provide to the participants in a batch exchange.  



Trading systems must not create or destroy any asset. This is the ``market-clearing'' property of Arrow Debreu exchange market equilibria. Formally:

\begin{axiom} [Asset Conservation]
\label{axiom:conservation}
Let exchange participants $i \in \agentset$ have endowments $\{x_{i,\asset{A}}\}_{\asset{A} \in \assetset}$ 
and receive bundles $\{z_{i,\asset{A}}\}_{\asset{A} \in \assetset}$. For each asset $\asset{A} \in \assetset$, we must have 
$\sum_{i \in \agentset} x_{i,\asset{A}} = \sum_{i \in \agentset} z_{i,\asset{A}}.$
%
\end{axiom}



The core fairness property of a batch trading scheme is that all orders in a batch trade at the same prices, and no trader gets an unfair advantage. Moreover, these prices are arbitrage-free. 
For this, the batch exchange must compute prices for all assets and ensure that no trader can get an allocation of a greater value than the value of their endowment.
  Formally:



\begin{axiom}[Uniform Prices]
\label{axiom:batch_valuations}
An equilibrium of a batch trading scheme has a shared market price $ \price{A}>0 $ for each asset 
$\asset{A}\in\assetset$.  
Let batch exchange participants $i \in \agentset$ have endowments $\{x_{i,\asset{A}}\}_{\asset{A} \in \assetset}$ and get allocated bundles $\{z_{i,\asset{A}}\}_{\asset{A} \in \assetset}$. For each participant $i \in \agentset$, 
$p\cdot z_i \leq p\cdot x_i.$
\end{axiom}

Note that we do not require strict equality between $p\cdot z_i$ and $p \cdot x_i$ in the axiom; instead, we require that no market participant should get a price strictly better than the equilibrium asset prices. However, in conjunction with asset conservation (Axiom~\ref{axiom:conservation}), uniform price (Axiom~\ref{axiom:batch_valuations}) implies that this inequality needs to be strict in equilibrium.

\begin{observation} \label{obs:same_prices}
    Let batch exchange participants $i \in \agentset$ have endowments $\{x_{i,\asset{A}}\}_{\asset{A} \in \assetset}$ and get allocated bundles $\{z_{i,\asset{A}}\}_{\asset{A} \in \assetset}$.  Axioms~\ref{axiom:conservation} and~\ref{axiom:batch_valuations} imply that for each participant $i \in \agentset$, 
$p\cdot z_i = p\cdot x_i.$
\end{observation}

\begin{proof}
    Summing the inequalities of Axiom~\ref{axiom:batch_valuations} over all participants, $\sum_{i \in \agentset} p \cdot z_i \leq p \cdot x_i.$ By asset conservation, 
    $\sum_{i \in \agentset} z_i = \sum_{i \in \agentset} x_i$. Since all $\price{A} > 0,$ all inequalities $p \cdot z_i \leq p \cdot x_i$ must be tight.
\end{proof}


A basic guarantee of batch trading systems in the literature is that a limit sell order is executed in full when the batch price exceeds the limit price and is not executed when the batch price is less than the limit price. When the batch price equals the limit price, the limit order trades any amount between zero and its maximum amount. We require that this condition be satisfied for all limit orders. Mathematically, this corresponds to the following:
\begin{axiom}[Best response trade for limit orders]
\label{axiom:br}
The allocation received by a limit sell order maximizes its linear utility function subject to the market asset prices.
\end{axiom}


The ``constant function market maker'' name may suggest that a CFMM shall trade from reserves $x$ to reserves $z$ only if $f(z)=f(x)$.  
One might consider
encoding this strict equality condition as an axiom; however, real-world CFMM deployments only check the 
weaker condition that $f(z) \geq f(x)$ \cite{uniswapswapcode}.  Keeping this flexibility allows us to have a much richer design space and is crucial to satisfy certain desirable properties of batch trading systems, as we shall see in the next section. 
%

\begin{axiom}[Non-decreasing CFMM trading function]
\label{axiom:cfmm_nondecreasing}
A CFMMs trading function does not decrease due to a trade in the batch, i.e., for a trade from state $x$ to $z$, we must have $f(z)\geq f(x)$.
\end{axiom}

We define market equilibrium as a solution which satisfies all Axioms~\ref{axiom:conservation},~\ref{axiom:batch_valuations},~\ref{axiom:br}, and~\ref{axiom:cfmm_nondecreasing}. Formally:

\begin{definition}[Market Equilibrium] \label{def:market_eq}
    For a batch trading in a set of assets $\assetset$, and market participants $i \in \agentset$ with initial endowments $\lbrace x_{i,\asset{A}} \geq 0\rbrace_{i \in \agentset, \asset{A}\in\assetset},$ a market equilibrium consists of a set of prices $\lbrace \price{A} > 0 \rbrace_{\asset{A}\in\assetset}$ and allocations $\lbrace z_{i,\asset{A}} \geq 0\rbrace_{i \in \agentset, \asset{A}\in\assetset}$ which satisfy Axioms~\ref{axiom:conservation},~\ref{axiom:batch_valuations},~\ref{axiom:br}, and~\ref{axiom:cfmm_nondecreasing} (asset conservation, uniform batch prices, limit orders get a best-response trade, and the CFMM trading functions are non-decreasing).
\end{definition}


\section{Desirable Properties of Market Equilibria}
\label{sec:desirables}

Market equilibria may not be unique, so we turn now to study desirable properties of market equilibria, which will then guide our design of market equilibrium computation algorithms.

\subsection{Pareto Optimality}~

Given the asset prices in equilibrium, the utility that a limit sell order receives is fixed per Axiom~\ref{axiom:br}. However, some market equilibria may have asset prices that provide a worse utility to the limit orders than other admissible market equilibria. We define Pareto optimality as follows.

\begin{definition}[Pareto Optimal Market Equilibria] \label{def:po}

For a batch instance, a market equilibrium $E$ is Pareto optimal if there does not exist another market equilibrium $E'$ which Pareto dominates $E$ for the utility of the limit sell orders.
\end{definition}

Note that in a market with only limit sell orders, \textit{all} market equilibria are Pareto optimal -- this property is lost if CFMMs also participate. There is no natural notion of the utility of a CFMM in our model, and since CFMMs are expected to charge trading fees (we do not model the fees in this paper), we define the Pareto optimality with respect to the utility of the limit orders only.

We illustrate here with an example that multiple market equilibria may exist even for simple instances, and many of those may not be Pareto optimal.

 \begin{example}[Pareto optimal market equilibria]
     Consider an instance of a batch exchange trading in assets $A$ and $B$, with a CFMM $f(x) = x_A x_B,$ in the state $x_A = 1, x_B=4,$ and a limit order $l = (A,B,1,1)$ i.e., for selling up to $1$ unit of $A$ for $B$ with a minimum price of $1$.
     
     The set of market equilibria is given by the asset prices $(p_A = r, p_B=1)$ for $r \in [1,2].$ The corresponding bundle that the limit order receives is $(z^l_A = 0, z^l_B = r),$ and that the CFMM receives is $(z^c_A = 2, z^c_B = 4-r).$
     The unique Pareto optimal market equilibrium corresponds to $r=2.$


 \end{example}

Pareto optimality safeguards the interests of the limit order based traders. We now give desirable properties aimed at safeguarding the interests of the CFMM liquidity providers.

\subsection{Price Coherence and Mitigating Cyclic Arbitrage}

For a CFMM trading standalone (not in a batch exchange), its state changes only in response to trade requests and not directly in response to the prices of assets in the external market. When the external market prices change, the CFMM's spot price is now ``stale.'' 
An arbitrageur can make a risk-free profit by buying from the CFMM some units of the asset whose relative price has increased on the external market and selling it there at the new (higher) price. 

An arbitrage opportunity also arises when a group of CFMMs are mispriced with respect to each other. %
Towards this, we define ``cyclic arbitrage'', which motivates important design considerations and properties we desire from market equilibria in batch exchanges.

\begin{definition}[Cyclic Arbitrage of a Group of CFMMs] \label{def:ca}
    A  group of CFMMs $\mathcal{C}$ is in a state of cyclic arbitrage if it is possible to obtain a non-zero amount of any asset for free by trading with the group $\mathcal{C}$ simultaneously.
\end{definition}

We now define a related property on the spot prices of CFMMs, which provides a handle on the analysis of cyclic arbitrage.

\begin{definition}[Price Coherence of a Group of CFMMs] \label{def:pc}
    A group of CFMMs $\mathcal{C}$ has price coherence if there exists a set of asset prices $\lbrace q_A>0\rbrace_{A\in\mathcal{A}}$ such that for each asset pair $(A,B) \in \mathcal{A}^2,$ every CFMM $C \in \mathcal{C}$ that trades in asset pair $(A,B)$ has a spot price $\frac{q_A}{q_B}$ units of B per unit of A.\footnote{When the CFMM trading function is not differentiable, we require that $\frac{q_A}{q_B}$ be in the set of subgradients of the level curve of the trading function, and all results regarding PC will continue to hold.}
\end{definition}

Price coherence precisely characterizes cyclic arbitrage in CFMMs.

\begin{proposition} \label{prop:ca-pc}
    A group of CFMM are in a state of cyclic arbitrage if and only if they are not in a state of price coherence.
\end{proposition}

\begin{proof}
Follows from the ``no-arbitrage condition'' of \citet[pg. 121]{angeris2022optimal}. They give a convex optimization framework trading with a group of CFMMs. In the case of no fee, for the objective of maximizing the sum of quantities of assets one can get without giving any asset to the CFMMs in net, they show that the objective value is zero if and only if there exists a global set of prices for
the assets consistent with the spot prices of the CFMMs.
\end{proof}

Regarding the operation of our batch exchange, we want the CFMMs participating in the exchange to be in a state of price coherence after each batch is executed. 

\begin{definition}[Price Coherence of Market Equilibria (PC)]
A batch exchange market equilibrium satisfies PC if the group of participating CFMMs have price coherence post-batch.
\end{definition}

We also introduce a weaker condition than PC for equilibrium, which might be more applicable in situations where PC is too strict.

\begin{definition}[Preservation of Price Coherence (PPC)] \label{def:ppc}
A batch exchange satisfies PPC if, for any batch instance with  
    a group of CFMMs in a state of price coherence pre-batch, the market equilibria computed by the exchange satisfy PC.
\end{definition}

PPC safeguards the CFMMs, but it comes at a cost. Towards this, we recall and prove Theorem~\ref{thm:pc-po}.
\begin{theorem*}[\ref{thm:pc-po} restatement]
A batch exchange cannot simultaneously guarantee Pareto optimality for limit orders and preservation of price coherence (PPC) for CFMMs.
\end{theorem*}

We prove this theorem by constructing examples of batch exchange instances
demonstrating the claim.  The examples demonstrate the intuition behind the theorem, but there is nothing special about the precise details of the construction.  Informally, it is possible for most sets of CFMMs (using an LCR) to find limit orders on which a batch exchange cannot guarantee PO.  Instead, we characterize the precise cases where this tradeoff does not occur in Theorem \ref{thm:univ3-ppc}.

\begin{proof}
   Consider a batch instance trading in assets $\mathcal{A} = \{A,B\}.$ There are two CFMMs: 
    
    $C_1$ with trading function $f_1(x) = x_Ax_B$ and pre-batch state $(x_A = 1, x_B=2)$ with spot price $2$, and
    
    $C_2$ with trading function $f_2(x) = x_A^2x_B$ and pre-batch state $(x_A=1,x_B=1)$ with spot price $2$.

    There is a limit sell order $(A,B,1,1)$, i.e., to sell up to 1 unit of A for B for a price of at least 1.

    The set of market equilibria is given by asset prices $(p_A = r, p_B = 1)$ for $r\in [1,(8+\sqrt{10})/9]$ and corresponding allocations.

The unique Pareto optimal market equilibrium corresponds to $r = (8+\sqrt{10})/9 \approx 1.24.$
   In this equilibrium, the post-batch spot price on CFMM $C_1$ is $\approx 0.769$ and that on CFMM $C_2$ is $\approx 0.748$.
\end{proof}

\subsection{Joint Price Discovery and Mitigating Parallel Running }

We motivated PC and PPC, intending to ensure that the group of CFMMs participating in the batch exchange do not end up in a state of cyclic arbitrage. However, the problem of arbitrage in financial systems is not limited to the trades that can be made with a group of CFMMs. 

Consider the following form of arbitrage.

\begin{definition}[Parallel Running] \label{def:pr}
For a batch instance, a parallel running arbitrage opportunity exists if for any real numbers $\tilde a, \tilde b >0$ and any assets $A, B \in \mathcal{A},$
    a trader can sell $\tilde a$ units of  asset $A$ in exchange for $\tilde b$ units of asset $B$ in the batch and, can then obtain $\tilde a$ units of asset $A$ in exchange for $\hat b < \tilde b$ units of asset $B$  by trading with the participating CFMMs post-batch.
\end{definition}


Parallel running is a similar concept to front running, which corresponds to making a trade based on advanced knowledge of future orders. 
By definition, precisely identifying parallel running opportunities requires knowledge of all the other batch participants and the equilibrium computation algorithm. However, estimating the chances and magnitudes of such opportunities with only partial information may be possible -- we leave this investigation for future work. We show here that we can design equilibrium computation algorithms that eradicate all parallel running opportunities, even if the arbitrager has perfect information about the batch.

First, we describe a property of market equilibria.

\begin{definition}[Joint Price Discovery (JPD)] \label{def:jpd}
   Let group $\mathcal{C}$ of CFMMs participate in a batch. A market equilibrium with asset prices $\lbrace p_A>0\rbrace_{A\in\mathcal{A}}$ satisfies JPD if for each asset pair $(A,B) \in \mathcal{A}^2,$ every CFMM $C \in \mathcal{C}$ that trades in pair $(A,B)$ has a post-batch spot price $\frac{p_A}{p_B}$ units of B per unit of A.
\end{definition}

In other words, JPD corresponds to the case where the post-batch CFMM spot prices are the same as the batch prices.
Observe that JPD is a particular case of PC. We use this name for this property since markets `discover' new prices as a result of trades, and JPD ensures that the participating CFMMs also discover the prices which emerge as a result of trading on the batch exchange.

The following result establishes the importance of JPD as a defence against parallel running.

\begin{lemma}
    Joint price discovery (JPD) makes parallel running impossible.
\end{lemma}

\begin{proof}
  See that in equilibrium, each order in the batch executes at the batch prices given by the ratios of asset prices $\lbrace p_A>0\rbrace_{A\in\mathcal{A}}.$ By quasi-concavity of CFMM trading function $f$, for any trade, the price obtained is no better than the spot price. By JPD, the spot price is equal to the batch price, so parallel running is impossible.
\end{proof}


We show in Appendix~\ref{sec:u_only_prf} that for a reasonable regularity condition ``split invariance'' on equilibrium computation algorithms (in Definition~\ref{def:split_indiff}) and for ``large'' batch instances (as in Definition~\ref{def:large_batch}), JPD is a \textit{necessary} condition to eradicate parallel running in batch exchanges. 

\subsection{Locally Computable Response for CFMMs} \label{sec:lcr}

Agents often need predictability in their trade in a batch, subject to the batch prices. While this is guaranteed for limit orders axiomatically (Axiom~\ref{axiom:br}), it would also be good for liquidity providers who invest their capital in CFMMs for market making.

\begin{definition}[Locally Computable Response (LCR)] \label{def:lcr}
A CFMM's trading rule in a batch exchange satisfies LCR if it is a map 
$F(x, f, p)\rightarrow z,$
where 
$\lbrace x_A\rbrace_{A\in\mathcal{A}}$ is the pre-batch CFMM state, $f$ is its trading function, and $\lbrace p_A\rbrace_{A\in\mathcal{A}}$ is 
a set of equilibrium asset prices. The output
$\lbrace z_A \rbrace_{A\in \mathcal{A}}$ is the post-batch CFMM state. 
 $F(x, f, p)$ must be invariant under rescaling of $p$.
\end{definition}

LCR provides predictability to liquidity providers and makes participation more lucrative (apart from the fees they charge, which are not considered in our model). However,
as with other desirable properties, we lose some design space if we impose LCR for CFMMs. For example, we have the following impossibility result.

\begin{proposition}\label{prop:pwtv}
   The price-weighted trade volume of a limit sell order is $||p \cdot |x-z| ||_1$ where $x$ and $z$ are their endowment and equilibrium allocation, respectively, and $p$ are the batch price.

No batch exchange satisfying LCR for CFMMs can guarantee a market equilibrium that maximizes the sum of the price-weighted trade volumes of the limit sell orders. 
\end{proposition}

\begin{proof}
Consider a batch exchange with a single CFMM with trading function $f(x) = x_A x_B$ and pre-batch state $x_A = 1, x_B = 4.$ We study two batch instances:
\begin{enumerate}[leftmargin = 18pt]
    \item There is one limit sell order $(A,B,3,1)$ i.e., for selling up to $3$ units of $A$ for a minimum price of $1$ $B$ per $A$.
    All equilibria have the same (up to rescaling) asset prices: $p_A = p_B = 1$. 
    The ``price-weighted trade volume of limit orders'' is maximized by the equilibrium at which the CFMM buys $3$ units of $A$ from the limit order.
    \item There is one limit sell order $(A,B,3,1)$ and another limit sell order $(B,A,3,1).$ 
    All equilibria have the same (up to rescaling) asset prices: $p_A = p_B = 1$. 
    The ``price-weighted trade volume of limit orders'' is maximized by the equilibrium at which the CFMM makes no trades.
\end{enumerate}
Since the equilibrium asset prices are the same in the two instances, any LCR trading rule for the CFMM cannot distinguish between the two instances and cannot optimize for the   
price-weighted trade volume of limit orders.
\end{proof}

 \citet{dfusionmodel} gives a mixed-integer program for the problem of finding a market equilibrium that maximizes the price-weighted trade volume of limit orders. 
 Finding polynomial-time algorithms for this objective or showing that none exist is an interesting open problem.
In the same vein as Proposition~\ref{prop:pwtv}, also recall the impossibility result of Theorem~\ref{thm:lcr-po}.
\begin{theorem*}[\ref{thm:lcr-po} Restatement]
A batch exchange cannot simultaneously guarantee Pareto optimality for limit orders and a locally computable response (LCR) for CFMMs.
\end{theorem*}

\begin{proof}
In the examples we construct, all CFMMs trade in two assets.
First consider the case of LCR CFMM trading rules which, for some starting state $x,$ trading function $f,$ and batch prices $p,$ demand an allocation $z'$ strictly ``above the curve'', i.e., $f(z') > f(x).$ Define $z$ as $z = \sup\limits_{\zeta~\in~\{\hat z|p\cdot(\hat z-x)~=~0;~f(\hat z)~\geq~f(x)\}} \|p\cdot~|\zeta-x|~\|_1$.

Consider a batch with one market sell order selling $z_A-x_A$ units of $A$ for $B.$ The unique Pareto optimal market equilibrium is where the CFMM attains state $z,$ which contradicts its LCR rule. 

Now, consider the case of LCR CFMM trading rules which, for all starting states $x,$ trading functions $f,$ and batch prices $p,$ demand allocation $z$  ``on the curve'', i.e., $f(z) =f(x).$  For 2-asset CFMMs, by the strict quasi-concavity of trading function $f$, either the CFMM makes no trade (i.e., $z$ = $x$) or demands allocation $z$ such that $z = \sup\limits_{\zeta~\in~\{\hat z|p\cdot(\hat z-x)~=~0;~f(\hat z)~\geq~f(x)\}} \|p\cdot~|\zeta-x|~\|_1$. Recall that this corresponds to Trading Rule S. We consider both these cases.
\begin{enumerate}[leftmargin = 18pt]
    \item There exists a CFMM $C$ which trades in assets $A$ and $B$ in the batch with an LCR trading rule which, for some starting state $x,$ trading function $f,$ and batch prices $p$ (not equal to the CFMM spot price), demands an allocation $z$, where $z = x$ (that is, makes no trade).

    Consider a batch where no other group of CFMMs together trades in $A$ and $B$. Consider a batch instance with a market sell order for $\hat a>0$ units of asset $A$ for $B$. The equilibrium under the above-stated LCR is sub-optimal for the limit order than receiving any non-zero amount of $B$.

    \item Consider a case where all CFMMs in the batch are two-asset CFMMs, and they follow the LCR Trading Rule S. Consider the following example with two CFMMs and two limit sell orders.

    CFMM $C_1$ trades in assets $A$ and $B,$ has $f(x) = x_Ax_B$ and a starting state  $(\hat x_A = 1, \hat x_B = 4).$

    CFMM $C_2$ trades in assets $B$ and $D,$ has $f(x) = x_Bx_D$ and a starting state  $( \hat x_B = 1, \hat x_D = 4).$

Limit order $L_1 = (A,D,3,1)$ i.e., to sell up to $3$ units of $A$ for $D$ at a minimum price of $1$ D per A. 

Limit order $L_2 = (A,B,1,1.5).$ 

On using Trading Rule S for both CFMMs, a market equilibrium $E$ has asset prices
$p_A=p_B=p_D=1$.
 CFMM $C_1$ buys $3$ units of $A,$ sells $3$ units of $B,$ and ends up in a state of $x_A = 4$ and $x_B = 1$.
 CFMM $C_2$ buys $3$ units of $B,$ sells $3$ units of $D,$ and ends up in a state of $x_B = 4$ and $x_D = 1$.
 Limit order $L_1$ trades in full and limit order $L_2$ does not trade.

  Another market equilibrium, $\hat E$, has asset prices $p_A=2; p_B = 1; p_D=2$. 

 CFMM $C_1$ buys $1$ unit of $A$ and sells $2$ units of $B$. CFMM $C_2$ makes no trades.   Limit order $L_1$ makes no trade and limit order $L_2$ trades in full (sells $1$ unit of $A$ for $2$ units of $B$).
 The utility of $L_1$ is the same in $E$ and $\hat E$, whereas that of $L_2$ is strictly higher in $\hat E$. Therefore market equilibrium $\hat E$ is a Pareto improvement over $E$. \qedhere
\end{enumerate}
\end{proof}

\section{Algorithms to Achieve Desirable Properties in Market Equilibrium}
\label{sec:pos}
These negative results in the previous section illuminate the fundamental trade-offs a batch exchange must operate under. In this section, we focus on the positive results. We design algorithms that obtain some of the desirable properties of the market equilibrium. For this, we adopt a viewpoint from the perspective of the CFMMs. Any deterministic algorithm for equilibrium computation can be described as a \emph{trading rule} for a CFMM, a function which specifies the CFMM's allocation, given all information of the batch instance.


\begin{definition}[CFMM Trading Rule] \label{def:tr}
A  CFMM's {\it Trading Rule} in a batch exchange is a map 
$F(x, f, p, \Gamma )\rightarrow z,$
where 
$\lbrace x_A\rbrace_{A\in\mathcal{A}}$ is the pre-batch CFMM state, $f$ is its trading function, $\lbrace p_A\rbrace_{A\in\mathcal{A}}$ is 
a set of equilibrium asset prices, and $\Gamma$ is all other inputs to the equilibrium computation algorithm (including information of all other agents). The output
$\lbrace z_A \rbrace_{A\in \mathcal{A}}$ is the post-batch CFMM state.
\end{definition}

Consider the following natural CFMM trading rule, which views it as a utility-maximizing agent.


\begin{definition}[Trading Rule U] \label{def:tru}
    With its trading function as a pseudo-utility function, give the CFMM a pseudo-utility-maximizing bundle of assets subject to the asset prices. 
    $$F_U(x, f, p) =  \sup\limits_{\zeta \in \{\hat z~|~p\cdot (\hat z -x) = 0\}} f(\zeta).$$
        Under Assumption~\ref{ass:cfmm} on CFMM trading functions, the allocation under Trading Rule U is unique.
\end{definition}
An illustration of Trading Rule U is in Figure~\ref{fig:rules}.
Notice that Trading Rule U satisfies LCR. That is, the CFMM's allocation in equilibrium interacts with the rest of the batch only via the asset prices $p$. When the CFMM demand under Trading Rule U satisfies an additional condition -- Weak Gross Substitutability (WGS)
\footnote{A demand function satisfies WGS if for all asset prices $\{p_A\}_{A \in \mathcal{A}},$ on decreasing the price $p_{A^*}$ of an asset $A^*$ while keeping all other prices $\{p_A\}_{A \in \mathcal{A}\setminus A^*}$ constant, the demand of all  assets other than $A^*$ does not decrease. WGS of all agents is a sufficient condition for the existence of market equilibria \cite{kuga1965weak} in Arrow-Debreu exchange markets.}
\footnote{Not all CFMM trading functions correspond to WGS demand functions under Trading Rule U -- many natural classes of trading functions do, for example, those given by monomials. On the flip side, some seemingly natural trading functions, for example, that of \textit{Curve} \cite{egorov2019stableswap}, do not. 
} --  Trading Rule U in the batch exchange can be implemented by well-known algorithms for computing equilibria in Arrow Debreu markets, such as the T\^atonnement-based algorithm described in \citet{codenotti2005market},
or our convex program of  \S\ref{sec:cvx}.
The simple Trading Rule U also has other desirable properties besides computational tractability.

\begin{proposition*} [\ref{prop:u-jpd} Restatement]
    A batch exchange has JPD if and only if it uses Trading Rule U for all CFMMs. 
\end{proposition*}
\begin{proof}
   The CFMM trading function $f$ is strictly quasi-concave. Thus, to maximize $f$ on a hyperplane (as in Trading Rule
U) is to find the point where the gradient of $f$ is normal to the hyperplane. This point is unique under Assumption~\ref{ass:cfmm}. Since the gradient of $f$ at
a point is equal to the spot price at said point (up to rescaling), 
setting the spot prices equal to the batch prices corresponds to Trading Rule U.
\end{proof}

Recall that JPD is a particular case of PC -- batch exchanges implementing Trading Rule U for all CFMMs also achieve PC. Adopting Trading Rule U protects CFMMs from both cyclic arbitrage and parallel-running-based arbitrage. The natural interpretation of Trading Rule U as treating CFMMs as utility-maximizing agents may be normatively significant in many scenarios.
%

We also study Trading Rule S, which is applicable only for batches with 2-asset CFMMs\footnote{For clarification, here we consider the cases where a CFMM trades in only two assets, but each CFMM may be trading in an arbitrary pair of assets, and the entire batch trades in an arbitrary number of assets.}. It maximizes the price-weighted trade volume of the CFMM, given the batch prices.
\begin{definition}[Trading Rule S] \label{def:trs}
    Maximize the CFMM's price-weighted trade size subject to the non-decreasing trading function constraint. Under Assumption~\ref{ass:cfmm}, the allocation is unique.
$$F_S(x, f, p) =  \sup\limits_{\zeta \in \{\hat z~|~p\cdot (\hat z -x) = 0;~f(\hat z) \geq f(x)\}} \|p\cdot|\zeta-x|\|_1.$$

\end{definition}

An illustration of Trading Rule S is in Figure~\ref{fig:rules}. Informally, it corresponds to trading ``all the way'' up to the point where the trading function non-decreasing constraint is tight. It always ends up on the same level curve of the trading function as the initial state. For the extreme sparse case of a single CFMM and a single limit order, Trading Rule S mimics a standalone CFMM, and therefore, in this case, a limit order based-trader does not strictly prefer trading ``outside the batch'' with the CFMM.
Notice that Trading Rule S also satisfies LCR, i.e., it gives a CFMM's ``demand'' as a response to asset prices. 
All CFMM trading functions lead to a WGS demand under Trading Rule S,
and as such, the equilibrium computation problem is computationally tractable
(for example, once again, using the algorithm of \citet{codenotti2005market}
or our convex program of \S \ref{sec:cvx}).

Although per Theorem~\ref{thm:lcr-po}, when guaranteeing LCR for all CFMMs, the batch exchange cannot guarantee Pareto optimal equilibria, it can nonetheless do so by using Trading Rule S for the special case of batches trading in only two assets when the CFMMs are price coherent pre-batch.

\begin{proposition} \label{prop:s-pareto-2asset}
 For batches trading in only two assets, if the CFMMs in the batch are price-coherent pre-batch, the equilibrium obtained by implementing Trading Rule S for all CFMMs is Pareto optimal.   
Trading Rule S is the only LCR CFMM trading rule with this property.
\end{proposition}

Proposition~\ref{prop:s-pareto-2asset}, in conjunction with the impossibility result of Theorem~\ref{thm:lcr-po}, shows how multi-asset batch exchanges pose much more analytical challenges than two-asset batch exchanges. While achieving Pareto optimality is possible under a natural LCR trading rule for two-asset batches (when pre-batch PC holds), it is impossible in the general case. Achieving Pareto optimal outcomes even when compromising on LCR is a non-trivial open problem.

Trading Rule S, in general, does not satisfy PPC. However, this negative result is bypassed in batches with only ``Concentrated Liquidity Constant Product'' (CLCP) CFMMs.

\begin{definition}[Concentrated Liquidity Constant Product (CLCP) Trading Functions] \label{def:clcp}
   A trading function in the CLCP class is $f(x)= (x_A+\hat{x}_A)(x_B+\hat{x}_B)$
for constants $\hat{x}_A > 0, \hat{x}_B > 0.$ 
\end{definition}

This class of CFMMs allocates liquidity in an interval of prices and implements the constant-product trading function in that interval. This price interval is $(\frac{\hat x_B^2}{K}, \frac{K}{\hat x_A^2})$ for constant $K$ denoting the initial value of the  CFMM trading function.  This class includes the constant-product CFMM \cite{uniswapv2} (where liquidity is spread to all prices, $0$ to $\infty$, by setting $\hat x_A = \hat x_B = 0$) and the constant-sum CFMM (where liquidity is concentrated at a single price). Combinations of CLCP trading functions form the basis of the widely successful decentralized exchange protocol Uniswap V3 \cite{uniswapv3}. Recall  the property of CLCP trading functions from Theorem~\ref{thm:univ3-ppc}.

\begin{theorem*} [\ref{thm:univ3-ppc} Restatement]
CLCP is the unique class of CFMMs such that if all CFMMs in a batch belong to this class, then the market equilibria attained by implementing Trading Rule S for all CFMMs satisfy PPC.    
\end{theorem*}

\begin{hproof}
Let $g_C(p,q)$ denote the spot price of a CFMM $C$ under Trading Rule S for batch price $p$ and initial spot price $q$.
We show that $g_C(p,q)$ must satisfy \textit{involution} on spot prices -- for any batch price $p,$ if $g_C(p,q) = \hat q,$ then $g_C(p,\hat q) = q$. We further show that for PPC, a necessary condition is $g_C(1,q) g_C(1,1/q) = 1.$ A trivial condition for all CFMMs is $g_C(p,p) = p.$ We then show that the only two solutions to these conditions are $g_C(p,q) = q$ and $g_C(p,q) = p^2/q.$ We also show that these solutions satisfy the general form of the PPC condition. The first case corresponds to the constant sum CFMM, and the second corresponds to the rest of the CLCP class.
\end{hproof}

%

\section{A Convex Program for 2-asset WGS Demands}
\label{sec:efficient}

\label{sec:cvx}

Here we give a convex program that computes market equilibria in batch exchanges that incorporate CFMMs that
trade between two assets and use locally-computable trading rules that satisfy WGS. 
Alternatively, this program computes equilibria 
in Arrow-Debreu exchange markets where every agent's demand response satisfies WGS, and every agent
has utility in only two assets.
The program is based on the convex program of Devanur 
et al. \cite{devanur2016rational} for linear exchange markets.

The key observation is that 2-asset CFMMs satisfying WGS can be viewed as (uncountable) collections of agents with linear utilities and
infinitesimal endowments.  This correspondence lets us replace a summation over agents
with an integral over this collection of agents.  However, proving the correctness of our program requires direct analysis of the objective instead of the argument based on Lagrange multipliers used in \cite{devanur2016rational}.
The objective is smooth on the feasible region, and gradients are easily computable for many natural CFMMs.

\subsection{From a Demand Function to a Continuum of Limit Orders}
\label{sec:cvxamm}

Suppose that a batch participant is only interested in two assets $A$ and $B$.
We first correspond such an agent to a continuum of exchange market agents with 
linear utility functions that trade between $A$ and $B$.
Under LCR, an agent's demand function can only depend on the exchange rate between the two assets (that is, $r_{AB}=\frac{p_A}{p_B}$).
We can therefore define a function $S_{A}(\cdot)$ that gives, for any price $r_{AB}$, the amount of asset $A$ that
the agent sells (in net, relative to its initial endowment).

\begin{definition}[Supply]
\label{defn:cumulative_density}
Consider a batch participant with endowment $x$ that only trades between assets $A$ and $B$.
The \textit{Supply} of the participant of asset $\asset{A}$ at exchange rate $r_{AB}>0$ is 
the set 
$S_{\asset{A}}(r_{AB})=\lbrace \max(0, x_{\asset{A}}-z_{\asset{A}}) \rbrace$, where $z$ is an amount of $\asset{A}$
that could be held by the agent after a batch is settled.

For agents maximizing a utility function (i.e. limit sell orders), $z$ ranges over utility-maximizing allocations, and
for a CFMM using some trading rule, $z$ is the allocation specified by the trading rule.
\end{definition}

A supply function $S_A(\cdot)$ is \textit{monotonic} if, for any $r_{AB} <\hat r_{AB}$ and 
$z_1 \in S_A(r_{AB})$, $z_2\in S_A(\hat r_{AB})$, we have $z_1 \leq z_2$.

For CFMMs under Trading Rule U or S, for every price $r_{AB}$, it must be the case that either
$0\in S_{\asset{A}}(r_{AB})$ or $0\in S_{\asset{B}}(1/r_{AB})$.
In the rest of this section, it will be convenient to consider the supply function for each asset separately. 
 Specifically, for the purpose of the convex program,
we will assume that an exchange market consists of a set of supply functions, and that each supply function $i$
sells good $\asset{A}_i$ and buys $\asset{B}_i$ (and therefore write only $S_i(\cdot)$, when clear from context).

When a CFMM trades based on an LCR trading rule, $S_i(\cdot)$ outputs a single point.
Similarly, for an agent maximizing a utility function, $S_i(\cdot)$ always outputs a single point when
the utility function is strictly quasi-concave.  Many natural LCR CFMM trading rules, 
such as rules S and U
correspond to 
differentiable supply functions when CFMM trading functions are strictly quasi-concave.

Proving existence of equilibria requires that each $S_i(\cdot)$ has a closed graph (for Kakutani's fixed-point theorem).
This property holds for
rules S and U when
trading functions are quasi-concave.

\begin{definition}[Inverse Supply]
$S_i^{-1}(t)$ is the least upper bound on the set $\lbrace r\vert \max(S_i(r)) \leq t\rbrace$.
When the set is empty, $S_i^{-1}(t)=0$, and is $\infty$ when the set is unbounded.
\end{definition}

We make the following simplifying assumption in the rest of the discussion.
\begin{assumption}
\label{ass:functionalform}
Every Supply function $S_i(\cdot)$ is either a monotonic, differentiable, point-valued function
with $S_i(\infty) > 0$
or a ``threshold'' function of following form: for some constants $r_i\geq 0$ and $\hat S_i>0$, $S_i(r)=0$ for all $r<r_i$, $S_i(r)=\hat S_i$ for $r>r_i$,
and $S_i(r_i)=[0, \hat S_i]$.
\end{assumption}

The results below only require that an agent's net trading behaviour at equilibrium is expressible as a sum of finitely many supply functions
trading from one asset to another.  
Continuous, strictly quasi-concave trading functions give point-valued, differentiable supply functions, and linear utility functions (i.e. each limit sell order)
 give threshold supply functions. 
When supply functions are differentiable, we can define the marginal supply of a CFMM at each exchange rate.

\begin{definition}[Marginal Supply function]
The marginal supply function $s(r)$ of an agent selling $\asset{A}$ in exchange for $\asset{B}$ is $\frac{d S(r)}{dr}$.
\end{definition}

The marginal supply function represents the marginal amount of $\asset{A}$ that an agent sells at each price.
Informally, $s(r)$ denotes the size of a limit sell offer with minimum price $r$,
and an agent with supply function $S(r)$ behaves as a continuum of these marginal limit sell offers.

Supply functions are monotonic if and only if an agent's behaviour satisfies WGS.

\begin{lemma}
\label{lemma:nonnegative_wgs}
A CFMM with a differentiable $S_i(\cdot)$ satisfies WGS if and only if $S_i(\cdot)$ is monotonic.
\end{lemma}

\begin{proof}

If there is some $r$ such that $S(r)$ is strictly decreasing at $r=\price{A}/\price{B}$, 
then there is some $\price{B}^\prime$ such that
$\price{B}^\prime < \price{B}$ (so $r^\prime=\price{A}/\price{B} > r$), $S_i(r) > S_i(r^\prime)$, and $S_i(\cdot)>0$ on the interval $[r, r^\prime]$
(so the other asset's supply function is $0$ on this interval).
In other words, a decrease in the price of $\asset{B}$ causes the agent to sell less $\asset{A}$, which means that its demand for $\asset{A}$ increases and violates WGS.

Conversely, if $S_i(r)$ is always nondecreasing, then for every $r=\price{A}/\price{B}$ and $r^\prime=\price{A}/\price{B}^\prime$ 
with $\price{B}^\prime < \price{B}$ ($r^\prime > r$),
$S_i(r^\prime) \geq S_i(r)$.  Thus, demand for $\asset{A}$ never increases as the price of $\asset{B}$ falls.
\end{proof}

\begin{corollary}
    All two-asset CFMMs with trading functions satisfying Assumption~\ref{ass:cfmm} have a WGS demand function under Trading Rule S.
\end{corollary}

\begin{proof}
    Trading functions satisfying Assumption ~\ref{ass:cfmm} are nondecreasing in every asset.  
    At price $r$, the CFMM with endowment $x$ sells $S(r)$ units of $\asset{A}$ for $rS(r)$ units of $\asset{B}$.
    At $r^\prime > r$, the CFMM has sufficient budget to purchase $r^\prime S(r)$ units of $B$, and if it does not purchase at least this much,
    then the trading function cannot be nondecreasing from $(x_{\asset{A}} - S(r), x_{\asset{B}} + rS(r))$ to 
    $(x_{\asset{A}} - S(r), x_{\asset{B}} + r^\prime S(r))$
\end{proof}

\subsection{Convex Program}

\label{sec:cvxprogram}

We assume that a set of agents has utility functions that imply a set of supply functions that satisfy Assumption \ref{ass:functionalform}, with $S_i(\cdot)$ for $i\in \cvxcfmmset{}$ continuous and $S_i(\cdot)$ for $i\in \cvxloset$ threshold functions. These agents trade a set of assets $\assetset$. Variables $\lbrace \price{A}\rbrace$ denote the price of assets $\asset{A}\in \assetset$. 

The quantity $y_i/\price{A}_i$ denotes the amount of good $\asset{A}_i$ that supply function $i$ sells to the market, receiving $y_i/\price{B}_i$ units of $\asset{B}_i$.  By construction, at equilibrium, it must be the case that $y_i/\price{A}_i\in S_i(\frac{\price{A}}{\price{B}})$. 

Define $\beta_{i,r}(p)=\min(\price{A}_i, \price{B}_i r)$. Informally, $\beta_{i,r}(p)$ is the inverse best bang-per-buck for the marginal limit sell order at limit price $r$ of supply function $i$.

Finally, for continuous marginal supply functions,
define $g_i(t)$ to be $\int_0^{S_i^{-1}(t)} s_i(r)~\ln(1/r)~dr$.
For threshold supply functions, define $g_i(t)=\min(\hat S_i, t)\ln(1/r_i)$.


\begin{lemma}
\label{lemma:gisconcave}
$g_i(\cdot)$ is a concave function, and $-\price{A}_ig_i(y_i/\price{A}_i)$ is convex.
\end{lemma}

\begin{proof}
For continuous supply functions,
$$\frac{\partial}{\partial t}g_i(t) 
= \frac{\partial}{\partial t} \int_0^{S_i^{-1}(t)} s_i(r) \ln (1/r) dr 
= \frac{d(S_i^{-1}(t))}{dt} s_i(S_i^{-1}(t)) \ln \Big(\frac{1}{S_i^{-1}(t)}\Big) 
=-\ln (S_i^{-1}(t))
$$
The first equality is by the definition of $g_i(t)$, the second is by applying the Leibniz integral rule, and the third follows since $s_i(\cdot)$ is the derivative of $S(\cdot).$ 

Since $S_i^{-1}(t)$ is non-decreasing under WGS, the derivative of $g_i(\cdot)$ is a decreasing function and therefore $g_i(\cdot)$ must be concave (for $t\geq 0$).
Concavity clearly holds for threshold supply functions.
$-\price{A}_i g_i(y_i/\price{A}_i)$ is the perspective transformation of a convex function, so it is convex.
\end{proof}



Altogether, we get the following convex program:

\begin{theorem}
\label{thm:cvx}

The following program is convex and always feasible.  Its objective value is always non-negative.
When the objective value is $0$, the solution forms an exchange market equilibrium with nonzero prices,
and when such an equilibrium exists, the minimum objective value is $0$.

\begin{align}
\text{Minimize} & \sum_{i\in \cvxcfmmset} \price{A}_i \int_0^\infty  \!\! s_i(r)~\ln\Big(\frac{\price{A}_i}{\beta_{i,r}(p)}\Big)~dr
  +\sum_{i \in \cvxloset} \price{A}_i \hat S_i \ln \Big(\frac{\price{A}_i}{\beta_{i,r_i}(p)}\Big)
&- \sum_{i\in \cvxcfmmset\cup \cvxloset} \price{A}_i~g_i(y_i/\price{A}_i)
\tag{COP}\label{eq:cop}
\\
\text{Subject to} & \sum_{i \in \cvxboth :\asset{A}_i=\asset{C}} y_i = \sum_{i\in \cvxboth :\asset{B}_i=\asset{C}}y_i & \forall \asset{C}\in\assetset \tag{C1}\label{eq:c1}
\\
~&\price{C} \geq 1 & \forall \asset{C}\in\assetset \tag{C2}\label{eq:c2}
\\ 
~& 0 \leq y_i \leq \price{A}_i S_i(\infty) & \forall i\in \cvxcfmmset\cup \cvxloset. \tag{C3}\label{eq:c3}
\end{align}
\end{theorem}

\begin{proof}
Lemma \ref{lemma:cvx_and_feasible} shows the convexity and feasibility of the program.
Lemma \ref{lemma:cvx_nonneg_zero_at_eq} shows that the objective value is nonnegative, and is $0$ if and only if the optimal solutions 
satisfy $y_i/\price{A}_i \in  S_i(p)$ for all $i$. 
Given the mapping between agents in the Arrow-Debreu exchange market and supply functions in the convex program,
as in Assumption \ref{ass:functionalform}, each $y_i$ implies a transfer of $y_i/\price{A}_i$ units of $\asset{A}_i$
out of the corresponding agent's endowment 
and $y_i/\price{B}_i$ units of $\asset{B}_i$ into their allocation.  By construction of the supply functions,
these transfers give an optimal bundle of assets for each agent.
When an equilibrium exists, the optimal objective value is $0$.
\end{proof}

Lemma \ref{lemma:eq_exists} in the appendix shows that an equilibrium with positive prices exists using the Kakutani Fixed point Theorem, 
given standard mild assumptions (i.e. Assumption \ref{ass:nonsatiated} -- for each asset, there exists an agent who has non-zero utility for it).

\begin{lemma}
\label{lemma:cvx_and_feasible}
The program with objective~\ref{eq:cop} and constraints~\ref{eq:c1},~\ref{eq:c2}, and~\ref{eq:c3} is convex and feasible.
\end{lemma}

\begin{proof}
$\beta_{i,r}(\cdot)$ is concave, positive, and non-decreasing, and $\ln(\cdot)$ is concave and nondecreasing, so each term
$s_i(r)\ln(\price{A}_i/\beta_{i,r}(p))$ is convex, and the integral or sum of convex functions is convex.
Feasibility follows from setting $y_i=0$ for all $i$ and and $\price{C}=1$ for all assets $\asset{C}\in\assetset$.
\end{proof}




\begin{lemma}
\label{lemma:cvx_nonneg_zero_at_eq}

The objective value of the convex program is nonnegative, and is zero if and only if $y_i \in \price{A}_iS_i(\price{A}_i/\price{B}_i)$ for all $i \in \cvxcfmmset\cup \cvxloset$.

\end{lemma}

\begin{proof}

Consider the following three quantities.

\begin{enumerate}
    \item $E_1 = \sum_{i\in \cvxcfmmset} \price{A}_i \int_0^\infty  s_i(r)~\ln\left(\frac{\price{A}_i}{\beta_{i,r}(p)}\right)dr
+\sum_{i \in \cvxloset} \price{A}_i \hat S_i \ln \left(\frac{\price{A}_i}{\beta_{i,r_i}(p)}\right)$

\item $E_2 = \sum_{i\in \cvxcfmmset} \price{A}_i \int_0^{S_i^{-1}(y_i/\price{A}_i)} s_i(r)~\ln\left(\frac{\price{A}_i}{\beta_{i,r}(p)}\right)~dr
+\sum_{i \in \cvxloset} y_i \ln \left(\frac{\price{A}_i}{\beta_{i,r_i}(p)}\right)$

\item $E_3 = \sum_{i\in \cvxcfmmset\cup \cvxloset} \price{A}_i~g_i(y_i/\price{A}_i)$
\end{enumerate}

By construction, for all $r$, $\ln \beta_{i,r}(p) \leq \ln(r) + \ln(\price{B}_i)$, or equivalently,
$\ln(1/r) \leq -\ln(\beta_{i,r}(p)) + \ln(\price{B}_i)$.  Furthermore, note that by the constraint (\ref{eq:c1}),
for all assets $\asset{C}$, $\sum_{i:\asset{A}_i=\asset{C}} y_i\ln \price{C} = \sum_{i:\asset{B}_i=\asset{C}} y_i\ln \price{C}$.
Additionally, for all $i\in \cvxloset$, because $y_i \leq \hat S_i \price{A}_i$, it must be the case that $\min(\hat S_i, y_i/\price{A}_i) = y_i/\price{A}_i$.

Using these facts, we get
\begin{align*}
E_3 &= \sum_{i\in \cvxcfmmset\cup \cvxloset} \price{A}_i~g_i(y_i/\price{A}_i) \\
&= \sum_{i\in \cvxcfmmset} \price{A}_i \int_0^{S_i^{-1}(y_i/\price{A}_i)} s_i(r)~\ln(1/r)~dr
+\sum_{i \in \cvxloset} \price{A}_i \min(\hat S_i, y_i/\price{A}_i)\ln(1/r_i) \\
&\leq \sum_{i\in \cvxcfmmset} p_{A_i}\int_0^{S_i^{-1}(y_i/p_{A_i})}s_i(r)\left(-\ln \beta_{i,r}(p) + \ln (\price{B}_i)\right)~dr
+ \sum_{i\in \cvxloset} y_i \left( -\ln \beta_{i, r_i}(p) + \ln(\price{B}_i)\right) \\
&= \sum_{i\in \cvxcfmmset \cup \cvxloset} y_i \ln \price{B}_i
+ \sum_{i \in \cvxcfmmset} p_{A_i}\int_0^{S_i^{-1}(y_i/p_{A_i})}s_i(r)\left(-\ln \beta_{i,r}(p)\right)~dr
+ \sum_{i\in \cvxloset} y_i \left( -\ln \beta_{i, r_i}(p)\right)\\
&= \sum_{i\in \cvxcfmmset \cup \cvxloset} y_i \ln \price{A}_i
+ \sum_{i \in \cvxcfmmset} p_{A_i}\int_0^{S_i^{-1}(y_i/p_{A_i})}s_i(r)\left(-\ln \beta_{i,r}(p)\right)~dr
+ \sum_{i\in \cvxloset} y_i \left( -\ln \beta_{i, r_i}(p)\right) \\
&= \sum_{i \in \cvxcfmmset} p_{A_i}\int_0^{S_i^{-1}(y_i/p_{A_i})}s_i(r)\left(\ln \price{A}_i -\ln \beta_{i,r}(p)\right)~dr
+ \sum_{i\in \cvxloset} y_i \left( \ln\price{A}_i-\ln \beta_{i, r_i}(p)\right)  \\
&= E_2
\end{align*}

For any $i$, define $\hat{r}_i = \frac{\price{A}_i}{\price{B}_i}$.  

For any $r<\hat{r}_i$, $\price{B}_i/\beta_{i,r}(p) = 1/r$, and otherwise $\price{B}_i/\beta_{i,r}(p)=1/\hat{r}_i \geq 1/r$.  
As such, for any $i\in \cvxcfmmset$, 
$\int_0^{S_i^{-1}(y_i/\price{A}_i)} s_i(r)~\ln(1/r)~dr = \int_0^{S_i^{-1}(y_i/\price{A}_i)} s_i(r)~\ln(\price{B}_i/\beta_{i,r}(p))~dr$ if and only if $y_i/\price{A}_i \leq S_i(\price{A}_i/\price{B}_i)$. 

Similarly, for any $i\in \cvxloset$, 
$\price{A}_i \min(\hat S_i, y_i/\price{A}_i)\ln(1/r_i) = 
y_i \left( -\ln \beta_{i, r_i}(p) + \ln(\price{B}_i)\right)$ if and only if 
$y_i = 0$ if $r_i > \hat{r}_i$.  These conditions hold for each $i\in \cvxboth$ if and only if $E_2=E_3$.
Furthermore, observe that $\price{A}_i/\beta_{i,r}(p) \geq 1$ for all $r$, and is equal to $1$ for any $r > \hat{r}_i$.
Rearranging $E_2$ gives:

\begin{dmath*}
E_2 = 
\sum_{i \in \cvxcfmmset} \price{A}_i\int_0^{S_i^{-1}(y_i/p_{A_i})}s_i(r)\ln \Big( \frac{\price{A}_i}{\beta_{i,r}(p)}\Big)~dr
+ \sum_{i\in \cvxloset} y_i \ln \Big( \frac{\price{A}_i}{\beta_{i, r_i}(p)}\Big) \\
\leq \sum_{i\in \cvxcfmmset} \price{A}_i \int_0^\infty  s_i(r)~\ln\Big(\frac{\price{A}_i}{\beta_{i,r}(p)}\Big)~dr
+\sum_{i \in \cvxloset} \price{A}_i \hat S_i \ln \Big(\frac{\price{A}_i}{\beta_{i,r_i}(p)}\Big) \\
= E_1
\end{dmath*}

Observe that for any $i\in \cvxcfmmset$, $\int_0^{S_i^{-1}(y_i/\price{A}_i)}s_i(r)\left(\ln \price{A}_i -\ln \beta_{i,r}(p)\right)~dr
= \int_0^{\infty}s_i(r)\left(\ln \price{A}_i -\ln \beta_{i,r}(p)\right)~dr$
if and only if $y_i/\price{A}_i = S_i(\price{A}_i/\price{B}_i)$.  Additionally, for any $i\in \cvxloset$, if $r_i < \hat{r}_i$, then 
$y_i \left( \ln\price{A}_i-\ln \beta_{i, r_i}(p)\right) = \price{A}_i \hat S_i \ln \left(\frac{\price{A}_i}{\beta_{i,r_i}(p)}\right)$
if and only if $y_i = \price{A} \hat S_i$.  These conditions hold for each $i\in \cvxboth$ if and only if $E_1=E_2$.
As such, $E_1 \geq E_3$, and the inequality is tight if and only if $y_i/\price{A}_i \in S_i(\price{A}_i/\price{B}_i)$ for all $i\in \cvxboth$.
But the objective of the convex program is $E_1-E_3$, proving the theorem.
\end{proof}

\subsection{Rationality of Convex Program}
\label{sec:rationality}

The program of \cite{devanur2016rational} always has a rational solution.  Our program may not, given certain CFMMs.
However, rational solutions exist when CFMMs belong to the following class.

\begin{theorem}
\label{thm:rationality}

If the expression $\price{A}_i S_i(p_{A_i}/p_{B_i})$ is a piecewise-linear, rational function of $\price{A}_i$ and $\price{B}_i$ for all $i$,
on the range where $S_i(\cdot) > 0$ and $S_i(\cdot)$ is monotonic,
then the convex program has an optimal rational solution.
\end{theorem}

\begin{example}
\label{ex:densityfns}
The supply function of the constant product CFMM with reserves $(x_{\asset{A}}, x_{\asset{B}})$ under Trading Rule U  in a batch exchange is $\max(0, (rx_{\asset{A}} - x_{\asset{B}}) / (2r))$.

\end{example}

However, this convex program cannot always have rational solutions; in fact, there exist simple examples using
natural utility functions for which the program has only irrational solutions.

\begin{example}
\label{ex:lmsr_irrational}

There exists a batch instance containing one CFMM based on the logarithmic market scoring rule and one
limit sell offer that only admits irrational equilibria under Trading Rule U.
\end{example}

\section{Conclusion and Open Problems}
\label{sec:conclusions}

Constant Function Market Makers have become some of the blockchain ecosystem's most widely used exchange systems.
  Batch trading has been proposed and deployed to combat some shortcomings of decentralised and traditional exchanges.  
Different implementations in practice have taken substantially different approaches to how these two innovations should interact.
 
%
We develop an axiomatic framework and describe several desirable properties of the combined system. While several pairs of these properties cannot be guaranteed simultaneously, we are able to provide algorithms that achieve subsets of these properties.

Finally, we construct a convex program that computes equilibria on batches containing CFMMs that each
trade only two assets. For many natural classes of CFMMs, the objective of this convex program
is smooth, and the program has rational solutions.

\textbf{Open Problems:}  CFMM fees pose an important question for batch exchange deployments.  Fair compensation for liquidity provision in batch exchanges may require, for example, quantitative models and a careful analysis of trade data.
A convex program for exchange market equilibria for general WGS utilities (not just on two assets) continues to be an open problem. 

In this work, we find Pareto optimal solutions only in a special case of two-asset CFMMs, when the CFMMs are price-coherent pre-batch. Designing algorithms for finding Pareto optimal solutions in multi-asset batches in polynomial time is an important open problems. We show that such an algorithm will necessarily have to give up LCR for CFMMs (Theorem~\ref{thm:lcr-po}) which will result in unpredictable trade sizes for CFMMs given batch prices.

Iterative algorithms like T\^atonnement \cite{codenotti2005market}, if implemented naïvely, would require for each demand query an iteration over every CFMM. \citet{ramseyer2021speedex} use a preprocessing step to compute the aggregate response of a group of limit sell offers in logarithmic time; identifying a subclass of CFMM trading functions in which the aggregate demand response of a large group of CFMMs can be computed efficiently (in sublinear time) would be of practical use.

\begin{acks}
This work is supported by the Future of Digital Currency Initiative at Stanford University, the IOG Research Hub, and the Office of Naval Research, award number N000141912268.  Thanks to Sahasrajit Sarmasarkar for helpful discussions and technical assistance.
\end{acks}

\bibliographystyle{ACM-Reference-Format}
\bibliography{literature}

\appendix
\section{CFMM examples} \label{app:cfmm_examples}
\begin{example}[CFMM Trading Functions]
\label{ex:cfmms}
~
\begin{itemize}
	\item

	The widely-known decentralized exchange, Uniswap, uses the Constant Product CFMM \cite{uniswapv2},
	which sets $f(x)=\prod\limits_{A\in \mathcal{A}} x_A,$ where $\mathcal{A}$ is the set of assets that the CFMM trades in. 
 \item	The Constant Sum CFMM uses the trading function $f(x) = \sum\limits_{A\in \mathcal{A}} r_A x_A $ for set of assets $\mathcal{A},$ and positive constants $\{r_A\}_{{A\in \mathcal{A}}}$ which represent the prices of the respective assets.%

	\item
	The Logarithmic Market Scoring Rule \cite{hanson2007logarithmic} corresponds to a CFMM with trading function $f(x) = |\mathcal{A}|- \sum\limits_{A\in \mathcal{A}}e^{-x_A}$
	 \cite{univ3paradigm}.

\end{itemize}
\end{example}

\section{Path Independence Property} \label{sec:pi}

We here discuss a property of the batch exchange (and not of market equilibria like the previous properties), defined when there is a single limit order-based trader.
Blockchain systems often have restricted throughput capabilities, and it is important to design mechanisms which do not incentivize traders to submit multiple small orders instead of a single big order\footnote{It is a standard practice in finance that investors and traders with large orders submit their orders in smaller parts to receive a better average price since market liquidity at a time is limited \cite{alfonsi2010optimal,cartea2016closed}. Another possible reason to do so is that the traders do not want to disclose their private information, which is reflected in their order size \cite{garriott2020trading}. We do not aim to restrict this practice. We wish to restrict the incentives for breaking down orders even when the available liquidity in the market does not change.}. Towards this, we define the path-independence property, which is inspired by the standalone operation of CFMMs.

\begin{definition}[Path Independence] \label{def:pi}
    An arbitrary group of CFMMs $\mathcal{C}$ participates in two consecutive batches and does not modify between the batches. Suppose a single trader exists and wants to sell some units of an asset $A \in \mathcal{A}$ in exchange for asset $B \in \mathcal{A}$ via market sell orders (that is, limit sell orders with limit price zero). 
    
    In a path-independent batch exchange, the amount of asset $B$ they receive on splitting the units of $A$ they sell into the two batches is independent of the split.
\end{definition}

This definition trivially extends to multiple batches and not just two.
Standalone CFMMs satisfy path independence \cite{angeris2020improved}. However, it is surprisingly difficult to satisfy in batch exchanges, except for a narrow case of Theorem~\ref{thm:s-pi}.  
We restate and prove Theorem~\ref{thm:s-pi} here.

\begin{theorem}\label{thm:s-pi}
  Batch exchanges with Trading Rule S for CFMMs satisfy path independence when there is only one CFMM in the batch trading in only two assets.
\end{theorem}
\begin{proof}

Under Trading Rule S, $f(z) = f(x)$ where $x$ and $z$ are the CFMM's pre-batch and post-batch states, respectively. For any amount of $A$ that the market order trader sells, the amount of $B$ in the CFMM is specified by the trading function preservation and therefore the amount of $B$ that the trader receives is also only a function of the net amount of $A$ sold by the trader.
\end{proof}

The path independence property, although defined in a restricted setup, provides a distinction between Trading Rules U and S. Even in the simple setup of a single CFMM and a single limit-order-based trader, batch exchanges implementing Trading Rule U do not satisfy path independence, as shown in the following example.

\begin{example}[Trading Rule U does not satisfy path independence with one CFMM]
    Consider a batch instance with a single CFMM with $f(x) = x_A x_B$ and pre-batch state $\tilde x_A = 1, \tilde x_B = 4.$ 
    
    If a trader submits a market sell order for 2 units of $A,$  they receive $1.6$ units of $B.$ (The new CFMM state: $x_A = 3, x_B = 2.4$, has a spot price of $0.8$ B per A, which matches the batch price offered).

    If the same trader first submits a market sell order for $1$ unit of $A,$ they receive $4/3$ units of B for it. The new CFMM state is $x_A =2, x_B =8/3$. Now if they make another market sell order for $1$ unit of $A,$ they receive $2/3$ units of $B$ for it, with the final CFMM state of $x_A = 3, x_B = 2$.  The overall trade is $2$ A for $2$ B, which is strictly better than what they received on trading in full in one batch.
\end{example}

Although simple and intuitive in the case of 2-asset batches, the path independence property is surprisingly difficult to satisfy. Even Trading Rule S does not satisfy it when more than one CFMMs are present, even if the CFMMs are price-coherent pre-batch, as in the following example.

\begin{example}
      CFMM $C_1$ trades in assets $A$ and $B,$ uses $f(x) = x_Ax_B$ and starts with  $(x_A = 10, x_B = 1).$

    CFMM $C_2$ also trades in assets $A$ and $B,$ uses $f(x) = x_A + 10 x_B$ and starts with $(x_A = 10, x_B = 0).$

If a trader sells 1 B, they receive 10 A; CFMM $C_1$ makes no trade, and CFMM $C_2$ trades to state $(x_A = 0, x_B = 1)$.
Then, if they sell another unit of B, CFMM $C_1$ trades to final state $(x_A = 5, x_B = 2)$ and CFMM $C_2$ makes no trade. Overall the trader receives 15 units of A for 2 units of B.

But if they sell 2B together, the price is (3/20) A per B;
CFMM $C_1$ trades to final state $(x_A = 1.5, x_B=0)$ and 
CFMM $C_2$ trades to final state  $(x_A=20/3, x_B=3/2)$.
The trader receives 40/3 A  for 2 B. 
\end{example}

This example shows that path independence is generally in conflict with the interests of the CFMM, and it can be satisfied only in restricted cases. Further investigation into the desirability and achievability of path independence in batch exchanges in left for future work.


\section{Beyond Market Equilibria Solutions}\label{sec:beyond-market}

Here, we study the case where the CFMMs can post-process their state after the batch trade and before being opened to standalone operation or the next batch.
CFMMs can restore price coherence through a careful post-processing step. This can, for example, be facilitated by the batch operator. A simple method is to run a ``phantom batch'' with no limit orders and an equilibrium computation algorithm implementing Trading Rule U for all CFMMs. By doing so, the batch exchange may retain some desirable properties of solutions that do not guarantee PC in the \textit{main batch} (for example Trading Rule S), and then also attain PC before subsequent operations by running the phantom batch. This 2-step solution seems to violate our impossibility result of Theorem~\ref{thm:pc-po}, but it does not! This solution violates the axiom of uniform prices (Axiom~\ref{axiom:batch_valuations}) since the asset prices in the phantom batch may deviate from those in the main batch. Therefore, the overall trades received by the CFMMs will be different from that implied by the main batch asset prices.
  Our second result in this section is even more interesting from a practical perspective, in the sense that the CFMMs will not need to trade at potentially bad prices in the post-processing step. For this, we study a class of LCR trading rules with the following property.
  
  Batches which implement a trading rule from this class for all CFMMs assure that the CFMMs can \emph{locally} post-process and attain a state of JPD. Importantly, this local post-processing strategy does not entail changing the CFMM trading function. 
It simply consists of capturing a part of the CFMM reserves as `profit' and removing it from the market-making pool. In the process, the CFMM reaches a state where the value of the trading function is equal to its pre-batch value.
The caveat is that this 2-step process violates asset conservation since assets are effectively removed from the system at the end of it. here we explain these trading rules for 2-asset CFMMs but it can easily be extended to CFMM which trade in multiple assets.

We first define Trading Rules E and F.

 For a CFMM trading in assets A and B, under Assumption~\ref{ass:cfmm}, there is an injective map from the spot price to the CFMM state. Motivated by the `rebalancing' strategy of \citet{milionis2022automated}, consider a CFMM trading rule which buys as much of asset A as it would hold `on the curve' when the spot price equals the batch price.

 \begin{definition}  \label{def:a}
  For the initial CFMM state $x$ and its trading function $f$, denote the map from the spot price $p$ to the amount of asset A in the CFMM state by $\mathbb{A}(x,f,p),$ i.e.,
      $$\mathbb{A}(x,f,p) = \{z_A | f(z) = f(x); \frac{\partial f}{\partial x_A}(z)/\frac{\partial f}{\partial x_B}(z) = p\}.$$ 
 \end{definition}

 \begin{definition}\label{def:b}
 For the initial CFMM state $x$ and its trading function $f$, denote the map from the spot price $p$ to the amount of asset B in the CFMM state by $\mathbb{B}(x,f,p),$ i.e., 
     $$\mathbb{B}(x,f,p) = \{z_B | f(z) = f(x); \frac{\partial f}{\partial x_A}(z)/\frac{\partial f}{\partial x_B}(z) = p\}.$$  
 \end{definition}

For strictly quasiconcave CFMM trading functions, $\mathbb{A}(x,f,p)$ and $\mathbb{B}(x,f,p)$ are singleton. For non-strict quasiconcave CFMM trading functions, we can use any arbitrary elements of the sets defined above for the rest of the discussion.
This enables us to define two CFMM trading rules.

\begin{definition} [Trading Rule E] The allocation obtained by a CFMM in equilibrium is such that it gets as much A as it holds on the pre-batch level curve of $f$ when the spot price equals the batch price. That is:
    $F_E(x, f, p) = z^E$ where $z^E_A = \mathbb{A}(x,f,p_A/p_B)$ and $z^E_B = x_B + (p_A/p_B)(x_A -\mathbb{A}(x,f,p_A/p_B)) .$
\end{definition}

 Similarly, we also have Trading Rule F, where the roles of the assets are reversed.

\begin{definition} [Trading Rule F]    The allocation obtained by a CFMM in equilibrium is such that it gets as much B as it holds on the pre-batch level curve of $f$ when the spot price equals the batch price. That is:
    $F_F(x, f, p) = z^F$ where $z^F_A = x_A + (p_B/p_A)(x_B -  \mathbb{B}(x,f,p_A/p_B))  $ and $z^F_B = \mathbb{B}(x,f,p_A/p_B).$

\end{definition}
    

Observe that both Trading Rules E and F satisfy LCR. We illustrate these trading rules in Figure~\ref{fig:rules}.
We define a class of LCR trading rules for 2-asset CFMMs -- \emph{Strict-Surplus} Trading Rules -- as a parameterized interpolation between Trading Rules E and F.

\begin{definition} [Strict-Surplus Trading Rules]  \label{def:sstr}
 The allocation obtained by a CFMM in market equilibrium under Strict-Surplus Trading Rule with parameter $\theta \in [0,1]$ is  $F_\theta(x, f, p) = z^{\theta}$ where
  $z^{\theta}_A = \theta z^E_A + (1-\theta) z^F_A~$ and,
  $~z^{\theta}_B = \theta z^E_B + (1-\theta)z^F_B.$
\end{definition}

Strict-Surplus Trading Rules have a special property captured in the following theorem.

\begin{theorem} \label{thm:brf-postprocess}
      A CFMM with trading function $f$ and initial state $x$ trading in assets A and B with a Strict-Surplus Trading Rule with parameter $\theta$ can, post-batch, extract a non-negative surplus $e^{\theta}$ where 
      
    $e^{\theta}_A = z^{\theta}_A - \mathbb{A}(x,f,p_A/p_B),~$ and    
    $~e^{\theta}_B = z^{\theta}_B - \mathbb{B}(x,f,p_A/p_B),$
    
    and reach a state $z^{\prime}$ where $z^{\prime}_A =  \mathbb{A}(x,f,p_A/p_B)$ and $z^{\prime}_B =  \mathbb{B}(x,f,p_A/p_B)$. 

      After surplus extraction, the trading function becomes equal to its pre-batch value, i.e., $f(z^{\prime}) = f(x)$.

    If all CFMMs in the batch extract their respective surplus, the final state satisfies JPD.

\end{theorem}

\begin{proof}
 By the strict quasiconcavity of the trading function, both Trading Rules E and F trade "above" the curve, and that $z_B^E > \mathbb{B}(x,f,p_A/p_B)$ and $z_A^F > \mathbb{A}(x,f,p_A/p_B)$. This implies that the surplus $e^{\theta}$ is non-negative.
JPD is ensured by the definitions of $\mathbb{A}(x,f,p_A/p_B)$ and $ \mathbb{B}(x,f,p_A/p_B)$.
\end{proof}

This result augments our design space significantly for implementing JPD beyond simply applying Trading Rule U. Informally, Trading Rule F, which is a Strict Surplus trading rule, may provide more liquidity to the market than Trading Rule U in some cases. Observe from Fig~\ref{fig:rules} that in this example, Trading Rule F implies a larger trade for the CFMM than Trading Rule U at the same batch price.
\section{Trading Rule U and Parallel Running} 
\label{sec:u_only_prf}
We show here that in the case where a market is not especially sparse, Trading Rule U is the only trading rule that eliminates parallel running.
We say that a batch is sufficiently \textit{large} if there is a market sell order trading between every asset pair that can be traded on 
a CFMM in the batch.

\begin{definition}[Large Batch] \label{def:large_batch}
 A batch instance has a set $\mathcal{C}$ of CFMMs. Let CFMM $c \in \mathcal{C}$ trade in the set of assets $\assetset_c.$ Denote the set of asset pairs that can be traded on some CFMM in $\mathcal{C}$ by  $Pairs(\mathcal{C}) = \{(A, B) \vert \exists c \in \mathcal{C}, \text{s.t.~} (A, B) \in \assetset_c\}.$ 
 
 We say that a batch instance is \textit{large} if there exists a market sell order for every pair in $Pairs(\mathcal{C})$. 
\end{definition}

We first define a property of equilibrium computation algorithm which is very natural but is required for our result in this section.
\begin{definition}[Split Invariance] \label{def:split_indiff}
    Consider two batch instances.
\begin{enumerate}
    \item There is a set $\mathcal{C}$ of CFMMs and a set $\mathcal{L}$ of limit or market sell orders.
    \item There is a set $\mathcal{C}$ of CFMMs and a set $\mathcal{L} \cup \lbrace \tilde{l}_1, \tilde{l}_2 \rbrace \setminus \lbrace l\rbrace $ 
    of limit or market sell orders, where $l = (\asset{A},\asset{B},k,r) $, $ \tilde{l}_1 = (\asset{A},\asset{B},t,r)$ and $ \tilde {l}_2 = (\asset{A},\asset{B},k-t,r)$, where $\asset{A}, \asset{B}$ are any two assets, $k>t>0$, and $r\geq 0$.
\end{enumerate}
    An equilibrium computation algorithm is \textit{split-invariant} if it always computes the same asset prices in the above two cases.
\end{definition}

Split invariance is a natural condition for algorithms whose output depend on the aggregate behavior of the limit orders, and never on the exact sizes of individual limit orders.  Many natural algorithms (such as \cite{codenotti2005market,bei2019ascending}) satisfy this condition.

We have the following result.

\begin{theorem}
\label{thm:jpd_necessary}
    In a batch exchange using a split-invariant equilibrium computation algorithm,
    JPD is necessary to ensure that parallel running is impossible on large batch instances.
\end{theorem}

\def\qab{\hat{q}_{\asset{A},\asset{B}}}
\def\pab{\hat{p}_{\asset{A},\asset{B}}}

\def\batch{\begingroup \mathcal{X} \endgroup}
\def\batchp{\batch ^ \prime}
\def\batchpp{\batch ^ {\prime \prime}}

\begin{proof}
Suppose that JPD is not guaranteed by a batch exchange, and that there is some large batch input $\batch$ where JPD is not satisfied.

As such, there must exist some CFMM $c$ in $\batch$ trading in some assets $\asset{A}$ and $\asset{B}$ 
for which the post-batch CFMM spot exchange rate $\qab $ differs from the equilibrium batch exchange rate $\pab=\frac{\price{A}}{\price{B}}$.
Without loss of generality, suppose that $\qab < \pab$, and denote 
$\delta_{\asset{A},\asset{B}} = \pab - \qab >0$.
Let $l$ be the market order that sells $A$ for $B$ in the batch, and denote the amount of $\asset{A}$ sold by this order by $s_{\asset{A},\asset{B}}$.

\def\epsab{\varepsilon_{\asset{A},\asset{B}}}

Define $\epsab$ to be the amount of $\asset{A}$ sold by CFMM $c$ as its spot price moves from $\qab$ to $\pab$ when operating as a standalone CFMM.  It must be the case that $\epsab> 0$ (or else the above scenario would not violate JPD).

Now, consider an alternate batch instance $\batchp$, which differs from $\batch$ only in that the order $l$ is replaced with two market orders
 $\tilde l_1$ and $\tilde l_2,$ of sizes $\epsab$  and  $s_{\asset{A},\asset{B}} - \epsab$ respectively.
Because the batch uses a split-invariant equilibrium computation algorithm, the exchange rate in $\batchp$ between $\asset{A}$ and $\asset{B}$ is the same
as that computed in $\batchp$, namely, $\pab$.

Now consider an additional, alternative batch instance $\batchpp$, which is equal to $\batchp$ except that it does not include order $\tilde{l}_1$.
We claim that $\batchpp$ admits a parallel-running opportunity.  A parallel runner (who has complete information of $\batchpp$ and the equilibrium computation algorithm) can insert a market sell order selling $\epsab$ units of $\asset{A}$ for $\asset{B}$, thereby creating batch $\batchp$.
After the batch, CFMM $c$'s spot price will be $\qab < \pab$, and the parallel runner can sell some amount of $\asset{B}$ less than $\epsab \pab$ 
to buy back $\epsab$ units of $\asset{A}$ from the CFMM.
\end{proof}

Note that Definition \ref{def:large_batch} is sufficient but not necessary for the proof of Theorem \ref{thm:jpd_necessary}.
Informally, many natural market scenarios admit similar parallel-running scenarios when not using JPD.

\section{Proof of Proposition~\ref{prop:s-pareto-2asset}}
\begin{proof}

Let assets $A$ and $B$ be traded on the exchange. 
Denote the market equilibrium obtained by implementing Trading Rule S for all CFMMs by $E$. 
By pre-batch price coherence, either all CFMMs sell $B$, or all sell $A$.
Without the loss of generality, say the CFMMs sell asset $B$ in equilibrium $E$. 

This implies that there exists at least one limit order which buys asset $B$. An alternate equilibrium with a higher price of asset $B$ cannot Pareto dominate equilibrium $E$ since then this limit order will have a lower utility.

Recall that the utilities of the limit orders are fixed given the price under Axiom~\ref{axiom:br}.
Keeping the price of $B$ constant, the utilities do not change, and we cannot find a Pareto improvement over $E.$

Consider lowering the price of $B$. Now we can hope to generate a Pareto improvement only when no limit order sells $B$, since any limit order selling $B$ will have a lower utility. 

For higher utility, the limit orders buying $B$ consume strictly more $B$ than in equilibrium $E$. However, the CFMMs, cannot sell any more of $B$ at a lower price by the quasi-concavity of the trading function $f$ and the fact that the CFMM is ``on the curve'' in Trading Rule S. Therefore such a utility improving market equilibrium is not possible.

With a simple example, we now show that Trading Rule S is the only LCR CFMM trading rule with this property. Consider a batch instance with a \textit{market} sell order for selling asset $A$ for $B$. Suppose for some initial state $x$, trading function $f$, and batch price $p$, the batch exchange does not implement Trading Rule $S$. Let the CFMM make trade $z-x$ under Trading Rule S at price $p$. Without loss of generality, $z_A - x_A > 0.$ Let the market sell order sell $z_A - x_A$ units of A. The unique Pareto optimal solution is to give $x_B- z_B$ units of B to the limit order by the strict quasiconcavity of $f$. Providing any more $B$ to the market sell order is infeasible and fewer is Pareto dominated by Trading Rule S. 
\end{proof}
\section{Proof of Theorem~\ref{thm:univ3-ppc}}
\begin{proof}
   Denote the asset prices which give the initial spot prices of the CFMMs by $\{q_A\}_{A \in \mathcal{A}}.$ These exist by pre-batch PC. Denote the batch equilibrium prices by $\{p_A\}_{A \in \mathcal{A}}$ and the asset prices which give the post-batch spot prices of the CFMMs by $\{\hat q_A\}_{A \in \mathcal{A}}.$ These must exist under PPC.
   Since each CFMM here trades in only two assets (Trading Rule S is defined only for 2-asset CFMMs), its initial spot price is $q_{AB} = q_A/q_B$ for assets $A$ and $B$ that it trades. Similarly, denote the batch price in $A$ and $B$ by $p_{AB}.$

For a given CFMM trading function, recall that Trading Rule S gives a map $F_S$ (Definition~\ref{def:trs}) from the initial state and the batch asset prices to the final state. Observe that this map depends on $\{p_A\}_{A \in \mathcal{A}}$  only via the batch price $p_{AB}.$ 
Since the CFMM trading function value is invariant under Trading Rule S, we can represent Trading Rule S as a map from the initial spot price and batch price to the final spot price. 
   Denote this map for trading function $f$ by $ g_f(q_{AB}, p_{AB}) \rightarrow \hat q_{AB}.$ 

   The map $g_f(\cdot,\cdot)$ must satisfy the following properties when we have PPC under Trading Rule S.

   \begin{enumerate}
       \item \textit{Reflexivity.} That is, $g_f(\rho,\rho) = \rho$ for all $\rho \in [0,\infty).$ This is required by the assumption of the structure of trading functions.
       \item \textit{Involution.} That is, $g_f(q,p) = \hat q \implies g_f(\hat q,p) = q$ for all $q,p,\hat q \in [0,\infty).$ This is because, under Trading Rule S, the CFMM returns to the initial state after two consecutive batches if the batch price is the same in both batches.
        \item \textit{PPC.} For a group of CFMMs, construct a graph with the CFMMs as nodes and an undirected edge between two nodes if they trade at least one common asset. For any cycle $\mathbb{C}$ of CFMMs in this graph, see that $\prod_{i \in \mathbb{C}} p_{A_iB_i} = 1$ by arbitrage-freeness of batch valuations (here CFMM $i$ trades in asset $A_i$ and $B_i$, $A_{i+1} = B_{i},$ and $B_{|\mathbb{C}|} = A_1$). Also $\prod_{i \in \mathbb{C}} q_{A_iB_i} = 1$ by pre-batch PC. PPC requires that $ \prod_{i \in \mathbb{C}} g_{f_i}(q_{A_iB_i}, p_{A_iB_i}) = 1.$
        
   \end{enumerate}

Let $\mathfrak{F}$ be the class of CFMM functions for which the result holds (that is, if all $f_i \in \mathfrak{F}$, then PPC holds under Trading Rule S). We are interested in finding this class of CFMMs. In particular, since the batch can have arbitrary CFMMs, the PPC condition mentioned above should hold for any collection of CFMMs from class $\mathfrak{F}$.

   We first use a weaker constraint -- that PPC must hold when all $f_i$ are identical, and there are only two CFMMs. For the solution obtained for this relaxed problem, we show that the general form of the PPC constraint is also satisfied for any number of CFMMs.

The relaxed form of the constraints are as follows:
\begin{enumerate}
    \item $g_f(\rho,\rho) = \rho$ for all $\rho \in [0,\infty).$ [Reflexivity]
    \item $g_f(q,p) = \hat q \implies g_f(\hat q,p) = q$ for all $q,p,\hat q \in [0,\infty).$ [Involution]
    \item $g_f(q,1) \cdot g_f(1/q,1) = 1.$ [Special case of PPC]
\end{enumerate}

   From Claim \ref{main_claim}, we obtain that the only continuous functions that satisfy $g_f(q,1) = \hat q \implies g_f(\hat q,1) = q $ for all $q, \hat q$ and $g_f(q,1) \cdot g_f(1/q,1) = 1$ for all $q$ are $g_f(q,1) = q$ and $g_f(q,1) = \frac{1}{q}.$

   Further, the condition of of reflexivity implies that for general $p$, the only solutions are $g_f(q,p) = q$ and $g_f(q,p) = \frac{p^2}{q}.$

   The case of $g_f(q,p) = q$ corresponds to the constant sum CFMM whose spot price is invariant to the CFMM state. Observe that this CFMM is in the CLCP class (Definition~\ref{def:clcp}) and corresponds the case where the liquidity is concentrated at one price.

   We now solve for the trading function corresponding to $g_f(q,p) = \frac{p^2}{q}.$

Let a CFMM trade in assets $A$ and $B$, and its reserves in each asset be $x_A$ and $x_B$. 
When CFMM trading function value $f(x)$ is invariant, it gives an injective map from $A$ in the reserves to $B$ in the reserves. Denote this map by $B_f(x_A)$ for a fixed level curve of the CFMM trading function $f$. In this notation, the spot price of the CFMM is $\frac{dB_f}{dx_A}.$

   For two points on a level curve, $(x_{A0},x_{B0})$ and $(x_{A},x_{B}),$  the condition   $g_f(q,p) = \frac{p^2}{q}$ translates to
   \begin{equation}
       \frac{dB_f}{dx_A}\bigg\rvert_{x_{A0}}  \cdot  \frac{dB_f}{dx_A} = \left(  \frac{x_{B} - x_{B0}}{x_{A} - x_{A0}} \right)^2 
   \end{equation}
   Taking the point $(x_{A0},x_{B0})$ as fixed and denote $ \frac{dB_f}{dx_A}\Big\rvert_{x_{A0}}$ by constant $-c$ for $c>0$. The differential equation yields the following indefinite integral.
     \begin{equation}
     -c \int \frac{dB_f}{(x_{B1} - x_{B0})^2 }   = \int   \frac{dx_A}{(x_{A1} - x_{A0})^2}
   \end{equation}
   Further solving gives 
      \begin{equation}
     \frac{-c}{(x_{B} - x_{B0}) }   =  \frac{1}{(x_{A} - x_{A0})} + c_1,
   \end{equation}
   where $c_1$ is the constant of integration. For $x_B \neq x_{B0}$ and $x_A \neq  x_{A0}:$
     \begin{equation}
         -c(x_{A} - x_{A0})   = (x_{B} - x_{B0}) + c_1(x_{A} - x_{A0})(x_{B} - x_{B0})
   \end{equation}
   Upon rearrangement of the terms, this yields the CLCP trading function form.

   So far, we have shown that membership in the CLCP class is a necessary condition for having the required property in the theorem statement. 
   
   We show that it is also sufficient.

   Recall the functional form of the CLCP class: $f(x) = (x_A+\hat x_A) (x_B + \hat x_B).$
   Wlog, on a level curve where $f(x) =1,$ the spot price is given by:

   $$ q = \frac{1}{(x_A + \hat x_A)^2}.$$

   The batch price when the CFMM state moves from $x_1$ to $x_2$ is $ p = -\frac{x_{B1} - x_{B0}}{x_{A1} - x_{A0}}.$

   Observe that $x_{B1} - x_{B0} = \frac{1}{ (x_{A1}+\hat x_A) }  - \frac{1}{ (x_{A0}+\hat x_A) } = \frac{(x_{A0}- x_{A1})}{(x_{A0}+\hat x_A)(x_{A1}+\hat x_A)}  .$

Plugging this into the expression for the batch price $p,$ we obtain that $q_1 q_0 = p^2,$ where $q_1$ and $q_0$ are the spot prices at $x_1$ and $x_0$ respectively.

   Since all CFMMs are from the CLCP class, and since $\prod_i p_i = 1$ over any trading cycle by the uniform prices axiom on batch exchanges, and $\prod_i q_{0,i} = 1$ over any trading cycle by pre-batch PC, we also obtain $\prod_i q_{1,i} = (\prod_i p_i)^2 / (\prod_i q_{0,i}) = 1$ over any trading cycle, thus implying PPC.

   This completes the proof.
\end{proof}

\begin{claim}\label{main_claim}
      Consider any $g(\cdot) $ satisfying the following properties.
        \begin{itemize}
            \item $g(x) g(\frac{1}{x}) = 1$ 
            \item $g(g(x)) = x$.
              \item Given any closed and bounded interval $\mathcal{B}$ in $(0,\infty)$, the number of times $xg(x)-1$ changes sign is bounded. Note that the bound can be a function of the interval $\mathcal{B}$.

        \item The function $g(\cdot)$ is continuous.
        \end{itemize}

        Then the solution to $g(\cdot)$ can be either $g(x) = \frac{1}{x}$ or $g(x) = x$.
\end{claim}

\begin{proof}
    Observe that $g(1)=1$.
    
    We first show that if $x<1$, either $g(x)=x$ or $g(x)>1$.
    Suppose for contradiction that there exists $\hat{x}<1$ with $g(\hat{x})\neq \hat{x}$ and $g(\hat{x})<1$.
    Define $h(x)=x-g(x)$.  Then, without loss of generality,
    $h(\hat{x}) >0$ and $h(g(\hat{x})) <0$, so, by the intermediate value theorem,
    there exists $y$ with $h(y)=0$ and thus $g(y)=y$, with $y$ between $\hat{x}$  and $g(\hat{x})$.

    Furthermore, because $g(\cdot)$ is a continuous bijection (because $g(g(x))=x$), $g(\cdot)$ must be monotone.  As such, we must have that $g(\cdot)$ is decreasing at $\hat{x}<1$, which contradicts the fact that $g(1)=1$.

    Next, consider the following equation, which we will prove later.
    \begin{equation}{\label{eq:sup_set}}
        \{x | xg(x) >1, 0< x\leq 1\} = \emptyset
    \end{equation}

    Equation \eqref{eq:sup_set} shows that $xg(x)\leq 1$ for $x<1$.  Suppose for contradiction that there exists 
    $\hat{x}$ with $\hat{x}g(\hat{x})\neq 1$ and $g(\hat{x})\neq \hat{x}$, and let $\hat{y}=g(\hat{x})$.
    From the conditions of the claim, we get $g(\frac{1}{\hat{y}})=\frac{1}{g(\hat{y})}=\frac{1}{\hat{x}}.$  Furthermore, $\hat{y}>1$, so $\frac{1}{\hat{y}}g(\frac{1}{\hat{y}}) \leq 1$ by Equation \ref{eq:sup_set}.
But $\frac{1}{\hat{y}}g(\frac{1}{\hat{y}})=\frac{1}{\hat{x}\hat{y}}>1$,
as $\hat{x}\hat{y}=\hat{x}g(\hat{x})<1$ by assumption, which is a contradiction.
Thus, $xg(x)=1$ for all $x<1$, and all that remains is to prove Equation \eqref{eq:sup_set}.

Proof of Equation \eqref{eq:sup_set}: 

    Assume the following for contradiction.
            
            \begin{equation}{\label{eq:supremum_assumption}}
                \sup_{ x \in (0,1]} \{x | xg(x) >1 \} = x_0>0
            \end{equation}

            Because $xg(x)$ is continuous, it must be the case that $x_0g(x_0)=1$.

%
            Now Equation \eqref{eq:supremum_assumption} implies that 
            $$\forall \epsilon >0 \text{ } \exists x \in (x_0-\epsilon,x_0) \text{ s.t. } xg(x) >1$$

            Consider some $\epsilon>0$, choose $\tilde{x}(\epsilon) \in (x_0-\epsilon,x_0)$ satisfying $g(\tilde{x}(\epsilon))\tilde{x}(\epsilon)>1$, and denote $\tilde{y}(\epsilon) = g(\tilde{x}(\epsilon))$.  
            Thus, $g(\frac{1}{\tilde{y}(\epsilon)}) = \frac{1}{\tilde{x}(\epsilon)}$, which implies that $\frac{1}{\tilde{y}(\epsilon)} g(\frac{1}{\tilde{y}(\epsilon)}) <1$.
            Observe that $\frac{1}{\tilde{y}(\epsilon)}<\tilde{x}(\epsilon)$. This further implies that $xg(x)$ changes sign in interval $\left(\frac{1}{\tilde{y}(\epsilon)},\tilde{x}(\epsilon)\right)$. 

            Now observe that $\lim_{\epsilon \to 0} \tilde{y}(\epsilon) = g(\lim_{\epsilon \to 0} \tilde{x}(\epsilon)) = g(x_0) = \frac{1}{x_0}$ and thus, $\lim_{\epsilon \to 0} \frac{1}{\tilde{y}(\epsilon)} = x_0$. This limit on $\tilde{y}(\epsilon)$ and the statement in the previous paragraph imply that for every $\zeta>0$, we have $xg(x)$ changing sign in $(x_0-\zeta,x_0)$. This is a contradiction as number of times $xg(x)$ changes sign in any bounded interval containing $x_0$ is unbounded.

\end{proof}

\section{Existence of Positive Equilibrium Prices}

\begin{assumption}
\label{ass:nonsatiated}
For every asset $\asset{A}$ and every set of prices $p$,
if $p_{\asset{A}}=0$,
then there exists at least one agent who always has positive marginal utility for $\asset{A}$ (no matter how much $\asset{A}$
the agent purchases).
\end{assumption}

We also make a standard set of assumptions (e.g. as in \cite{arrow1954existence}) on utility functions of agents in an Arrow-Debreu exchange market.

\begin{assumption}
\label{ass:utility}
A utility function $u(\cdot)$ satisfies the following properties.
\begin{enumerate}
  \item[1] $u(\cdot)$ is continuous.
  \item[2] For all endowments $x$, there exists $x^\prime$ with $u(x) < u(x^\prime)$.
  \item[3] For any $x, x^\prime$ with $u(x) < u(x^\prime)$ and any $0 < t < 1$, 
  $u(x) < u(tx + (1-t)x^\prime)$.
\end{enumerate}
\end{assumption}

\begin{lemma}
\label{lemma:eq_exists}
If Assumption \ref{ass:nonsatiated} holds in some Arrow-Debreu exchange market consisting of CFMMs with supply functions
satisfying Assumption ~\ref{ass:functionalform} and agents with utility
functions satisfies Assumption \ref{ass:utility},
then there always exists an equilibrium of that market and every equilibrium has a positive price on every asset.
\end{lemma}

\begin{proof}
Let $\hat{e}=(\Sigma_ie_i) + (1,...,1)$ be a vector of assets strictly larger than the total amount of each asset available in the market,
and let $E=\lbrace x \vert 0 \leq x \leq \hat{e}\rbrace$.  Let $P$ be the price simplex on $\assetset$ 
(P=$\lbrace p\vert \Sigma_{\assetset} p_{\asset{A}} =1 \rbrace$).

Consider the following game.  There is one player for each agent in the exchange market and one ``market'' player,
for a combined state space of $E\times...\times E\times P$.  Clearly, this state space is compact, convex, and nonempty.

 Given the set of prices $p$,
 each agent player picks a utility-maximizing set of goods $x_i$
subject to resource constraints
(specifically, for each agent $i$, $x_i$ maximizes $u_i(\cdot)$ subject to $p\cdot x\leq p\cdot e_i$ and $x_i\in E$),
receiving payoff $u_i(x_i)$.
The market player chooses a set of prices $\hat{p}$, for payoff $\Sigma_i (e_i-x_i)\cdot p$.

Define $A_i(p)=\lbrace x\vert x\in E, p\cdot x \leq p \cdot e_i\rbrace$ be the set of utility-maximizing sets of goods
for agent $i$. Quasi-concavity of $u_i(\cdot)$ implies that $A_i(p)$ is convex, $A_i(\cdot)$ cannot be nonempty.
It suffices to show that $A_i(\cdot)$ is a continuous function for each agent $i$.

Let $p_1, p_2...$ and $x_1,x_2,...$ be any sequences of prices and demand responses converting to $p$ and $x$, respectively,
with $x_j\in A_i(p_j)$ for all $j\in\mathbb{Z}_+$.  Let $r_j=p_j\cdot x_j$ for each $j$.  Naturally, the sequence $r_1,r_2,...$ must
converge to $r=p\cdot x$, and the sequence $u_i(x_1),u_i(x_2)...$ must converge to $u_i(x)$.

Consider a sequence $\lbrace x_j^\prime\rbrace$ converging to $x^\prime$ with $p_j\cdot x_j^\prime=r_j$
and $x^\prime\in A_i(p)$.  It must be the case that $u_i(x_j)\geq u_i(x_j^\prime)$.  Because $\lbrace u_i(x_j^\prime)\rbrace$
converges to $u_i(x^\prime)$, we must have that $u_i(x)\geq u_i(x^\prime)$, so $x\in A_i(p)$ and thus $A_i(\cdot)$ must be continuous.  

If some agent is a CFMM defined by a supply function (instead of an agent maximizing a utility function), 
it suffices that the CFMM's action space, implicitly defined by the supply function, satisfies the same conditions.
These must hold for any CFMM Supply function satisfying Assumption ~\ref{ass:functionalform}.

It follows from Kakutani's fixed point theorem that there must exist a fixed point of this game; that is, there exists a $p$ and an $x_i$ for each $i$
such that $x_i\in A_i(p)$ and $(\Sigma_i x_i)_{\asset{A}} \leq 0$, with the inequality tight if $p_{\asset{A}}>0$ for all assets $\asset{A}$.

By Assumption \ref{ass:nonsatiated}, if $p_{\asset{A}}=0$ for some asset $\asset{A}$, then there exists an agent $i$ for which every $x_i\in A_i(p)$ has
$(x_i)_{\asset{A}} = \hat{e}_{\asset{A}}$.  
But then the demand for $\asset{A}$ must exceed the available supply, so $(\Sigma_i x_i)_{\asset{A}} >0$, a contradiction.
Analogously, it must be the case that for every asset $\asset{A}$ and every agent $i$, it must be the case that 
$(x_i)_{\asset{A}} < \hat{e}_{\asset{A}}$, and thus (by parts 2 and 3 of Assumption \ref{ass:utility}) 
$x_i$ maximizes $u_i(\cdot)$ subject to $x_i\cdot p\leq e_i$, $x_i\geq 0$ (i.e. without the restriction that $x_i\in E$).

\end{proof}

\section{Proof of Theorem~\ref{thm:rationality}}

\begin{proof}
Note that it suffices without loss of generality to consider functions only of the forms outlined in Assumption \ref{ass:functionalform}.
For each continuous $S_i(\cdot)$,
at an optimal point $(\tilde{p}, \tilde{y})$, for every CFMM $i$, it must be the case that 
$\tilde{\price{A}}_i S_i(\tilde{\price{A}}_i / \tilde{\price{B}}_i) = \tilde{y}_i$.
Then (because the expression $\price{A}_i S_i(p_{A_i}/p_{B_i})$ is piecewise linear, by assumption) there are
two linear functions $q_i^+(\cdot)$, $q_i^-(\cdot)$ such that,
on an open neighborhood about $\tilde{y}$,
$y_i = q_i^+(p_{A_i}, p_{B_i})$ when $\frac{\price{A}_i}{\price{B}_i} \geq \frac{\tilde{\price{A}_i}}{\tilde{\price{B}_i}}$ and 
$y_i = q_i^-(p_{A_i}, p_{B_i})$ otherwise.

We construct an alternate program with only linear constraints but the same optimal solution as the original program.
To the set of existing constraints in the convex program, add the constraints that
$y_i=q_i^+(\price{A}_i, \price{B}_i)$ and $y_i=q_i^-(\price{A}_i, \price{B}_i)$.
For each $S_i(\cdot)$ representing a threshold function of size $\hat S_i$ at exchange rate $r_i$,
if $\tilde{y}_i>0$, add the constraint 
$\price{A}_i \geq r_i \price{B}_i$,
and if $\tilde{y}_i < \tilde{\price{A}}_i \hat S_i$, add the constraint that $\price{A}_i \leq r_i \price{B}_i$.
Additionally, if $\tilde{\price{A}}_i > r_i \tilde{\price{B}}_i$, then add the constraint that 
$y_i =\price{A}_i \hat S_i$, and if $\tilde{\price{A}}_i < r_i \tilde{\price{B}}_i$,
add the constraint $y_i=0$.


This system of constraints is clearly satisfiable since $\tilde p, \tilde y$ is a solution, and every point satisfying the constraints is a market equilibrium (every point satisfies
$y_i \in \price{A}_i S_i(\price{A}_i / \price{B}_i)$).
Each of the constraints is linear and rational, so these constraints define a rational polytope.  The extremal points of this polytope must therefore
be rational.
\end{proof}
\section{Proofs of Examples in \S \ref{sec:cvx}}
\label{app:examples_rational}
\subsection{Proof of Example~\ref{ex:densityfns}}
\begin{proof}
Without loss of generality, assume batch price $r$ (units of $\asset{B}$ per $\asset{A}$) is greater than the CFMM's pre-batch spot price 
of $x_{\asset{B}}/x_{\asset{A}}$, so the CFMM in net sells $\asset{A}$ to the market and purchases $\asset{B}$. The CFMM, under Trading Rule U, makes a trade so that its reserves $(z_{\asset{A}}, z_{\asset{B}})$ after trading satisfy the following two conditions.  

First, the post-batch spot price of the CFMM, $z_{\asset{B}}/z_{\asset{A}}$, must be $r$.  
And second, the CFMM must trade at the batch exchange rate, so
$(x_{\asset{A}} - z_{\asset{A}})r = (z_{\asset{B}}-x_{\asset{B}})$. Note that $S(r)=x_{\asset{A}}-z_{\asset{A}}$.

Combining these equations gives
\begin{equation*}
\frac{x_{\asset{B}} + rS(r)}{x_{\asset{A}} - S(r)}=r.
\end{equation*}

Solving for $S(r)$ gives
\begin{equation*}
S(r) = \frac{rx_{\asset{A}} - x_{\asset{B}}}{2r}. \qedhere
\end{equation*}
\end{proof}

\subsection{Proof of Example~\ref{ex:lmsr_irrational}}

\begin{proof}

Consider a batch instance trading assets $A$ and $B$ that contains one CFMM and one limit sell order.
The CFMM uses trading function $f(x)=-(e^{-x_A}+e^{-x_B})+2$, with initial state $x_{A0}=x_{B0}=1$.  The limit sell order
is to sell 100 units of $A$ for $B$, with a minimum price of $\frac{1}{2}$ $B$ per $A$.

If the batchprice $p=p_{A}/p_B$ is strictly greater than $\frac{1}{2}$, then the limit sell order must sell the entirety of its $A$
to receive at least $100p > 50$ units of $B$, which the CFMM cannot provide.

On the other hand, if the batch price is less than $\frac{1}{2}$, then the limit sell order will not sell any $A$ but the CFMM
demands a nonzero amount of $A$.
Thus, at equilibrium, the batch price equals $\frac{1}{2}$.

Let the limit sell order sell $\tau$ units of $A$.  Clearly, in equilibrium, $0\leq \tau \leq 100$.

Furthermore, the spot price of the CFMM at equilibrium must be equal to $\frac{1}{2}$.  The spot exchange rate of this CFMM
is
\begin{equation*}
\frac{\frac{\partial f}{\partial x_A}}{\frac{\partial f}{\partial x_B}} 
= \frac{e^{-x_A}}{e^{-x_B}} = e^{-x_A + x_B}
\end{equation*}
Thus, at equilibrium, we must have that

\begin{equation*}
p = e^{-(x_{A0} + \tau) + (x_{B0} - p\tau)} = e^{-x_{A0} + x_{B0}}e^{-\tau-p\tau} = e^{-\tau(1+p)}.
\end{equation*}

This gives
 $e^{-\frac{3\tau}{2}}  = \frac{1}{2}.$
That is, $\tau=\frac{2}{3} \ln(2)$, which is irrational.
\end{proof}

\end{document}